\documentclass[journal,11pt,twocolumn]{IEEEtran}
\usepackage{graphicx}
\usepackage{cite}
\usepackage{array}
\usepackage{tabulary}
\newcolumntype{K}[1]{>{\centering\arraybackslash}p{#1}}

\usepackage{chemarrow}
\usepackage{color}
\usepackage{algorithmic}
\usepackage{algorithm}
\usepackage{clrscode}
\usepackage{amsmath, amssymb, mathrsfs, amsfonts, amsthm}
\usepackage{epsfig, epstopdf}
\usepackage{multirow}
\usepackage{bm}
\usepackage{enumitem}

\allowbreak
\usepackage{caption}
\usepackage{stfloats}
\usepackage{subcaption}
\usepackage{enumerate}
\usepackage{float}

\theoremstyle{theorem}
\newtheorem{lemma}{Lemma}

\newtheorem{remark}{Remark}

\newcommand{\V}{\mathrm{Var}}
\newcommand{\E}{\mathrm{E}}
\newcommand{\C}{\mathrm{Cov}}
\newcommand{\Prob}{\mathrm{P}}

\renewcommand{\arraystretch}{0.9}
\newcommand\blfootnote[1]{%
  \begingroup
  \renewcommand\thefootnote{}\footnote{#1}%
  \addtocounter{footnote}{-1}%
  \endgroup
}
\newtheorem{corollary1}{Corollary}[lemma]

\allowdisplaybreaks
\date{}
\title{On Molecular Flow Velocity Meters}
\author{\IEEEauthorblockN{
Maryam Farahnak-Ghazani, Mahtab Mirmohseni, and Masoumeh Nasiri-Kenari
}\\
\IEEEauthorblockA{Sharif University of Technology}}

\begin{document}
\maketitle
\begin{abstract}
Flow velocity is an important characteristic of the fluidic mediums. In this paper, we introduce a molecular based flow velocity meter consisting of a molecule releasing node and a receiver that counts these molecules. We consider both flow velocity detection and estimation problems, which are employed in different applications. For the flow velocity detection, we obtain the maximum a posteriori (MAP) decision rule. To analyze the performance of the proposed flow velocity detector, we obtain the error probability, its Gaussian approximation and Chernoff information (CI) upper bound, and investigate the optimum and sub-optimum sampling times accordingly. We show that, for binary hypothesis, the sub-optimum sampling times using CI upper bound are the same. Further, the sub-optimum sampling times are close to the optimum sampling times. For the flow velocity estimation, we obtain the MAP and minimum mean square error (MMSE) estimators. We consider the mean square error (MSE) to investigate the error performance of the flow velocity estimators and obtain the Bayesian Cramer-Rao (BCR) and expected Cramer-Rao (ECR) lower bounds. Further, we obtain the optimum sampling times for each estimator. It is seen that the optimum sampling times for each estimator are nearly the same. The proposed flow velocity meter can be used to design a new modulation technique in molecular communication (MC), where information is encoded in the flow velocity of the medium instead of the concentration, type, or release time of the molecules. The setup and performance analysis of the proposed flow velocity detector and estimator for molecular communication system need further investigation.\blfootnote{This work was supported by the Iran National Science Foundation (INSF) Research Grant on Nano-Network Communications and the Research Center of Sharif University of Technology.}
\blfootnote{The authors are with the Department of Electrical Engineering, Sharif University of Technology, Tehran, Iran (email: maryam.farahnak@ee.sharif.edu, \{mirmohseni, mnasiri\}@sharif.edu)}
\end{abstract}

\section{Introduction}
Measuring medium flow velocity is an important problem with many applications; in molecular communication (MC) (for finding the channel state information),
in industry (for abnormality detection),
or in health-care (for measuring the blood flow velocity).
The channel state information (CSI) is needed at the receivers of MC systems \cite{jamali2016channel2}. The parameters of the channel that specify the CSI in flow-assisted diffusion-based systems include the distance between the transmitter and the receiver, diffusion coefficient of the molecules, and the flow velocity of the medium. There are multiple works that consider distance estimation in MC, including \cite{huang2013distance, noel2014bounds, wang2015distance, lin2019high, gulec2020distance}. Joint estimation of the channel parameters is investigated in \cite{noel2015joint}, and estimation of the channel impulse response is considered in \cite{jamali2016channel2, rouzegar2017channel}. 
However, to the best of our knowledge, there is no work that specifically considers either flow velocity estimation or detection in MC.
Regarding the application in industry, flow meters are used in oil and gas industry for anomaly detection \cite{barbariol2019machine}. In healthcare, measuring blood flow velocity is important for monitoring heart function \cite{yusheng2018numerical}.
The classic flow meters are mechanical devices which 
have certain applications based on the passing fluid, whose velocity is intended to be measured \cite{IJET10568}.  
One of the important applications of flow meters is to measure the flow velocity of the blood. As stated above, blood flow velocity measurement is important in medical applications for monitoring heart function in order to diagnose cardiovascular or other vascular diseases \cite{yusheng2018numerical}. Some traditional blood flow velocity measurement methods are {indicator method} \cite{lassen2011indicator, puri2014radioisotope, zierler2000indicator}, {finger plethysmography} \cite{grote2003finger}, and {electromagnetic based method} \cite{blendis1970comparative}. In \cite{lassen2011indicator, puri2014radioisotope, zierler2000indicator}, the blood flow velocity is measured by injecting indicator molecules and using mass balance equation. In \cite{li2016effect}, skin temperature measurement after receiving acupuncture manipulations is used to measure the blood flow velocity. In \cite{grote2003finger}, the finger blood flow velocity is measured using finger plethysmography. These methods have low resolution. Methods with higher resolution, based on the microfluidic technology, are {ultrasonic doppler method} and {laser doppler method} \cite{yusheng2018numerical}. 

In this paper, we use a molecular transmitter-receiver setup to measure the medium flow velocity. The molecules, which exist in the medium or released from a molecular source, can provide significant information, for instance, to design the molecular flow meter. Since the medium flow velocity affects the concentration of the received molecules, the flow velocity can be measured by monitoring the concentration changes. To this end, a molecular receiver can be employed to sense the concentration of the received molecules, and measure the flow velocity. 
This resembles a MC structure where a transmitter releases some molecules and a receiver senses the concentration of these molecules. MC has advantages in mediums that are more compatible with bio or chemical molecules like the human body or environmental applications. MC systems have been studied from different aspects, e.g., system modeling \cite{
Chou2013, 
Einolghozati2013, 
 Kadloor2012, Arjmandi2013}
capacity analysis \cite{aminian1, fekri2, Arjmandi2014}, coding and modulation techniques \cite{Arjmandi2013, Leeson2012, 
Mahfuz2013, IWCIT2016type}, inter-symbol interference (ISI) mitigation techniques \cite{IWCIT2016physical, Noel2014improving, cho2017effective, Mosayebi2014, IWCIT2016type, Movahednasab2015, farahnak2018medium}, and channel estimation \cite{jamali2016channel2, noel2015joint, rouzegar2017channel}. 
The idea of sensing the concentration of molecules to measure the flow velocity was also used in the indicator method, which has been first introduced in 1824 by Hering to measure the blood flow velocity. In this method, some indicator molecules are injected 
to the blood vessel and sampled from other part of the vascular system. Then using the mass balance equation, the mean value of the blood flow velocity is measured. In other words, the mean flow velocity is measured as the change in the fluid volume per unit time. When the change in the concentration of the indicator molecules is fixed, the mean flow velocity is written as the change in the indicator quantity per unit time divided by the change in the concentration of the indicator molecules. Hence, in this method, the steady state behavior of the system is considered.
Further, the indicator method is just devoted to blood flow velocity measurement and is studied in physiology. In this paper, using a MC analysis setup, we introduce a molecular flow meter, which can measure the flow velocity in any fluidic medium with laminar flow, i.e., Raynold number less than $2100$ ($\textrm{Re}<2100$). Further, we assume that the movement of the molecules is affected by both flow velocity and diffusion, and none of them is fully dominated. Hence, we consider the advection-diffusion equations to obtain the average value of the received concentration. Then, we determine the flow velocity by applying the conventional detection or estimation methods. 
As an important application, this flow velocity meter can be used to design a new modulation method in MC, i.e., instead of encoding the information on the properties of the released signal (concentration, type, or the release time of the molecules), we can encode the information on the properties of the medium specifically on the medium flow velocity, and at the receiver, we can decode the information by detecting the medium flow velocity. This modulation method can have advantages on the prior methods in the sense of simplicity of the transmitter.

The degrees of freedom in designing the proposed flow meter include the sensing times of the molecular receiver, which need to be optimized for better performance of the flow meter. For performance investigation, different metrics, such as the time it takes to detect a change in the velocity, and the error probability of the flow meter can be considered. 
The samples taken at the receiver are statistically dependent in general, and obtaining the optimum sampling times is a challenging work. Further, the restrictions in some receivers, like Ligand receivers which have memories, make the problem more challenging. We remark that the medium flow velocity that we want to measure may be a random process, which either exists in the medium or is intentionally generated for communication purposes. 

We consider both flow velocity detection and estimation problems with different applications and assumptions. In some applications,
like designing the flow-based modulation in MC, which is described above, the flow velocity can belong to a finite set with a cardinality equal to the number of messages transmitted per channel use, and we need to design a detector to obtain the flow velocity, and hereby decode the message.
In this case, we assume $M$ hypotheses for the velocity ($M$ different functions of time and location) and use hypotheses testing methods, \cite{poor2013introduction}, to detect the function. 
In some applications, like finding the channel state information in MC, the medium flow velocity can take a real value in general and we need to design an estimator to obtain its value. In this case, we assume a constant flow velocity (both in time and location), which is chosen from a prior probability distribution function (PDF), and apply estimation methods, \cite{poor2013introduction}, to obtain the flow velocity. In this paper, we mostly focus on the flow velocity detection, and at the end of the paper, we briefly consider the flow velocity estimation.

The design of a general flow meter requires knowing the exact statistics of the medium and the existing molecules in the medium, i.e., how molecules are generated and propagated. To study the effect of certain parameters on the performance of the flow meter, we need to simplify the reality by adopting a simple model. Hence, we make a few assumptions and study the effect of sampling time on the performance of the flow meter. We assume that the source of molecules is a node that releases some fixed molecules in some specific time instances, and the receiver is a transparent receiver, \cite{gohari2016information}, with a sampling decoder, i.e., we assume that the receiver has a volume that counts the number of molecules inside its volume
at some time instances. 
We consider an $L$-sample receiver and further assume that the samples at the receiver are statistically independent, which can be achieved if the samples are taken with some time apart. Also, we assume that there is no boundary condition on the medium, since obtaining the channel impulse response of the medium with time variant flow is a challenging work in presence of boundary conditions in the medium. Our main contributions are as follows:
\begin{itemize}
\item{We design a molecular flow velocity meter, counting the number of arrived molecules affected by the flow of the medium. Our setup consists of a molecule releasing node and a receiver that samples the number of counted molecules. }
\item{For the molecular flow velocity {\em detecter}:\\
-- We obtain the maximum-a-posteriori (MAP) decision rule and for the one-sample decoder, we obtain the optimum threshold.\\
-- We consider the performance analysis of the proposed detector. For this purpose, we derive the error probability, its Gaussian approximation, and Chernoff information (CI) upper bound on the error probability.\\
-- We obtain the optimum and sub-optimum sampling times by minimizing the error probability, its Gaussian approximation, and the CI upper bound. For $M=2$, it is seen that the sub-optimum sampling times using CI upper bound are equal. For $M>2$, when the number of samples, $L\rightarrow \infty$, it is seen that the sub-optimum sampling times yield to ${M \choose 2}$ times with ${M \choose 2}$ weights. 
}
\item{For the molecular flow velocity {\em estimator}:\\
-- We obtain the MAP estimator and for the one-sample receiver, we obtain a closed-form estimator.\\
-- We obtain the minimum mean square error (MMSE) estimator and further simplify the equations for the linear MMSE (LMMSE) estimator.\\
-- We investigate the mean square error (MSE) of the above estimators. Further, we obtain the Bayesian and expected Cramer-Rao lower bounds on the MSE.\\
-- We obtain the optimum sampling times that minimize the MSE. When $L \rightarrow \infty$, it can be seen that the sampling times for the MAP estimator yield to two times with two weights.}

\end{itemize}


The structure of the paper is as follows: In Section \ref{sec1}, we describe the proposed molecular flow velocity detector/estimator setup. In Section \ref{sec2:flow_detector}, we consider the flow velocity detector and obtain the MAP decision rule, and derive its performance. Further, we obtain the optimum and sub-optimum sampling times. In Section \ref{sec3:flow_estimator}, we consider the flow velocity estimator and obtain the MAP 
estimator. Then, we obtain the estimation error and the optimum sampling times.
The numerical results are given in Section \ref{Simulation}. Finally, in Section \ref{conclusion} we conclude the paper.

\textbf{Notation}: Throughout the paper, vectors are shown with bold letters and their magnitudes, i.e., norm 2 of the vectors, are shown with non-bold letters.

\vspace{-0.2em}
\section{Molecular Flow meter} \label{sec1}
We propose a molecular flow velocity meter to measure the medium flow velocity in a flow-assisted diffusion-based system. 
To do this, we assume that there is a node at the origin, which releases some constant number of molecules in some time instances, and there is a molecular receiver in point $\bm{r}_0$, which receives these molecules and computes the medium flow velocity.  Hence, $\bm{r}_0$ is a vector which connects the releasing node to the receiver. We note the direction of this connecting line with $\bm{d}$ and its value with $r_0$. Hence, $\bm{r}_0=r_0 \bm{d}$ (see Fig. \ref{fig_sys_model}). 
The releasing node may have different behaviors. Assume $g(\bm{r},t)$ be the concentration of released molecules at point $\bm{r}$ and in time $t$. In the following, we mention some of the possibilities of the releasing node:
\begin{figure}
\centering
\includegraphics[trim={0.2cm 0 0 0cm}, scale=0.63]{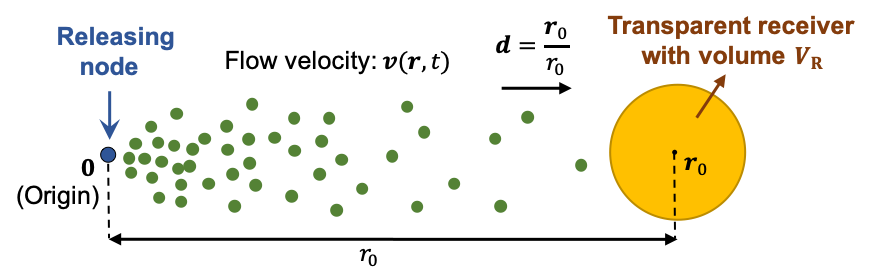}
\caption{The system model of the flow velocity meter}
\label{fig_sys_model}
\vspace{-1.5em}
\end{figure}

i) Burst releasing: a burst of molecules, noted by $\zeta$, is released at time instance $t=t_\textrm{r}$. For this releasing node we have $g(\bm{r},t)=\zeta \delta(t-t_\textrm{r})\delta(\bm{r})$.

ii) Pulse releasing: the molecules with rate $\gamma$ are constantly released starting at $t=t_\textrm{r}$. For this releasing node we have $g(\bm{r},t)=\gamma \delta(\bm{r})u(t-t_\textrm{r})$.

In this paper, we assume the {\em burst releasing} node.


\textbf{Channel model:} In the flow-assisted diffusion-based system, the movement of molecules is affected by both flow velocity and diffusion. Hence, the concentration of molecules at point $\bm{r}$ and in time $t$, noted by $c(\bm{r},t)$, in a medium with flow velocity $\bm{v}(\bm{r},t)$ satisfies the following advection-diffusion equation:
\begin{equation}\label{diff_eq}
\begin{aligned}
\frac{\partial}{\partial t}c(\bm{r},t)+\nabla. (\bm{v}(\bm{r},t)c(\bm{r},t))=D\nabla^2 c(\bm{r},t)+g(\bm{r},t),
\end{aligned}
\end{equation}
where $D$ is the diffusion coefficient of molecules.
When the medium flow is location invariant, i.e., $\bm{v}(\bm{r},t)=\bm{v}(t)$ (which means that the flow velocity is the same in every point of the medium and the change in the flow velocity of one point propagates to other points quickly), the advection-diffusion equation in \eqref{diff_eq}, reduces to:
\begin{equation}\label{difftvarflow}
\begin{aligned}
\frac{\partial }{\partial t}c(\bm{r},t)+\bm{v}(t).\nabla c(\bm{r},t)=D\nabla^2 c(\bm{r},t)+g(\bm{r},t).
\end{aligned}
\end{equation}

\textbf{Reception model:} We assume that the receiver is modeled by a sphere in 3-D with volume $V_\textrm{R}$ (radius $r_\textrm{R}$) and consider a transparent receiver, i.e., the receiver can perfectly count all molecules that fall into its volume. 
n
Denoting the mean number of received molecules in the $l$-th sample as $\Lambda_{l}$, the number of molecules counted by the receiver (noted by $Y_{l}$) follows a Poisson distribution with parameter $\Lambda_{l}$, i.e., $Y_{l} \sim \textrm{Poiss}(\Lambda_{l})$ \cite{gohari2016information}. If the radius of the receiver is sufficiently small compared to the distance between the transmitter and the receiver, $c(\bm{r},t)$ can be assumed uniform inside the receiver volume and $\Lambda_{l}$ can be approximated using the concentration of the received molecules at the origin of the receiver volume as $\Lambda_{l} \approx V_\textrm{R} c(\bm{r}_0,t_{l})$\cite{noel2014diffusive, noel2015joint}.
For simplicity of analysis, we assume that the observations $Y_1,...,Y_L$ are independent. This assumption can be made if the samples are taken with some time apart, i.e., the sampling times have sufficient distance from each other \cite{noel2014optimal}.


The impulse response of this system, $h(\bm{r},t)$, is defined as the concentration of molecules at point $\bm{r}$ and in time $t$, which is the solution of \eqref{difftvarflow}, for input $g(\bm{r},t)=\delta(\bm{r})\delta(t-t_\textrm{r})$. The channel impulse response when the medium flow velocity is location invariant (i.e., $\bm{v}(\bm{r},t)=\bm{v}(t)$) is obtained in \cite{Toshiaki2017} using Ito's calculus for the mean location of molecules, i.e., if $m(t)$ is the mean location of molecules, using Ito's calculus we have
$m(t)=\int_{t_\textrm{r}}^{t}\bm{v}(\tau)d{\tau}$. Hence, for 3-D diffusion,
$$h(\bm{r},t)=\frac{1[t>t_\textrm{r}]}{(4 \pi D (t-t_\textrm{r}))^{{3}/{2}}}\exp(-\frac{||\bm{r}-m(t)||^2}{4 D (t-t_\textrm{r})}).$$
The impulse response of the channel with time variant flow velocity can be written using the impulse response of the channel without flow velocity, noted by $h_0(\bm{r},t)$, as follows:
\begin{equation}\label{impulseresponse}
\begin{aligned}
h(\bm{r},t)=h_0(\bm{r}-\int_{t_\textrm{r}}^{t}\bm{v}(\tau)d{\tau},t),
\end{aligned}
\end{equation}
where
\begin{equation}\label{eq_h0}
\begin{aligned}
h_0(\bm{r},t)=\frac{1[t>t_\textrm{r}]}{(4 \pi D (t-t_\textrm{r}))^{{3}/{2}}}\exp(-\frac{||\bm{r}||^2}{4 D (t-t_\textrm{r})}).
\end{aligned}
\end{equation}

Now, using the channel impulse response we obtain $\Lambda_{l}$. Assuming the burst releasing node, we have $c(\bm{r},t)=\zeta h(\bm{r},t)$ and hence, $\Lambda_{l}= V_\textrm{R}c(\bm{r}_0,t_{l})=\zeta V_\textrm{R} h(\bm{r}_0,t_{l})$. Thus, according to \eqref{impulseresponse}, we obtain
\begin{align}\label{burstnode}
\Lambda_{l}=\zeta V_\textrm{R} h_0\big(\bm{r}_0-\int_{t_\textrm{r}}^{t_{l}}\bm{v}(\tau) d\tau,t_{l}\big).
\end{align}

\section{Flow velocity detector}\label{sec2:flow_detector}
Consider $M$ hypotheses $H_0,H_1,...,H_{M-1}$ corresponding to the flow velocities $\bm{v}_0(\bm{r},t)$, $\bm{v}_1(\bm{r},t),...,$ $\bm{v}_{M-1}(\bm{r},t)$. We let $I=\{0,...,M-1\}$ and denote the mean number of counted molecules at the receiver in sampling time $t_{l}$ for the hypothesis $H_i$ as $\lambda_{i,l}, i \in I, l=1,...,L$. For the location invariant flow velocity and the transparent receiver, from \eqref{burstnode}, we have $\lambda_{i,l}=V_\textrm{R}\zeta h_0(\bm{r}_0-\int_{t_\textrm{r}}^{t_{l}} \bm{v}_i(\tau) d\tau,t_{l})$.
 We assume that the prior probabilities of the hypotheses are equal, i.e., $\Prob(H_i)=\frac{1}{M}, i \in I$.

\begin{lemma}\label{lemma_threshold_L}
For a molecular flow velocity detector with $M$ hypotheses, and $L$-sample decoder at the receiver, the optimum MAP decision rule is obtained as
\begin{align}\label{decision_rule_csk}
\hat{i}=\arg\max_{i \in I} \sum_{l=1}^{L} y_l \ln (\lambda_{i,l})-\lambda_{i,l},
\end{align}
\end{lemma}
\begin{proof} Using MAP decision rule with equal prior probabilities for the hypotheses, we have
$$\hat{i}=\arg\max_{i \in I}\Prob(y_1,y_2,...,y_L|H_i).$$
For the independent observations, $\Prob(y_1,y_2,...,y_L|H_i)=\prod_{l=1}^{L} \Prob(y_l|H_i)$. The conditional probability distribution of $Y_l$ given $H_i$ assuming counting noise at the receiver is $\textrm{Poiss}(\lambda_{i,l})$ for $i \in I, l=1,...,L$. Hence, 
\begin{align}
\nonumber
\hat{i}&=\arg\max_{i \in I} \prod_{l=1}^{L} \frac{(\lambda_{i,l})^{y_l}\exp(-\lambda_{i,l})}{y_l!}\\\nonumber
&=\arg\max_{i \in I}\prod_{l=1}^{L} (\lambda_{i,l})^{y_l}\exp(-\lambda_{i,l})\\
&=\arg\max_{i \in I}\sum_{l=1}^{L}y_l \ln(\lambda_{i,l})-\lambda_{i,l}.
\end{align}
\end{proof}
\begin{corollary1}\label{corollary_lemma_threshold_L_1}
For 
binary hypothesis,
the optimum decision rule is simply obtained as 
\begin{align}\label{threshold2}
\sum_{l=1}^{L} w_l y_l \overset{H_0}{\underset{H_1}\gtrless}\beta,
\end{align}
where $w_l=\ln(\frac{\lambda_{0,l}}{\lambda_{1,l}})$ and $\beta=\sum_{l=1}^{L}(\lambda_{0,l}-\lambda_{1,l})$.
For $L=1$, the optimum decision rule is a simple threshold rule as $y_1 \overset{H_0}{\underset{H_1}\gtrless}\mathcal{T}$, with the threshold
$\mathcal{T}=\frac{\lambda_{0,1}-\lambda_{1,1}}{\ln\big(\frac{\lambda_{0,1}}{\lambda_{1,1}}\big)}.$
\end{corollary1}

\begin{lemma}\label{lemma_Pe_2}
The error probability in detecting the flow velocity, with $M$ hypotheses 
and $L$-sample decoder at the receiver, 
is obtained as follows:
\begin{equation}\label{Pe_Poisson}
\begin{aligned}
P_\textrm{e}&=1-\frac{1}{M}\sum_{i=0}^{M-1}\sum^{\infty}_{\substack{y_1,...,y_L=0, \\ \sum_{l=1}^{L} w_{i,j,l} y_l>\beta_{i,j},\\j\in I,~j \neq i}}\prod_{l=1}^{L} \frac{(\lambda_{i,l})^{y_l} \exp(-\lambda_{i,l})}{y_l!},
\end{aligned}
\end{equation}
where $w_{i,j,l}=\ln(\frac{\lambda_{i,l}}{\lambda_{j,l}})$ and $\beta_{i,j}=\sum_{l=1}^{L}(\lambda_{i,l}-\lambda_{j,l})$.
\end{lemma}
\begin{proof} The proof is provided in Appendix \ref{AppendixProoflemma2}.
\end{proof}

\begin{corollary1}\label{corollary_lemma_Pe_2_1}
For binary hypothesis, the error probability is simplified as:
\begin{align}\label{Pe_M_2}
\nonumber
P_\textrm{e}&=\frac{1}{2}\big[1-\sum_{\substack{y_1,...,y_L=0,\\ \sum_{l=1}^{L} w_l y_l>\beta}}^{\infty}\big(\prod_{l=1}^{L} \frac{(\lambda_{0,l})^{y_l} \exp(-\lambda_{0,l})}{y_l!}\\
&\quad-\prod_{l=1}^{L} \frac{(\lambda_{1,l})^{y_l} \exp(-\lambda_{1,l})}{y_l!}\big)\big],
\end{align}
where $w_l$ and $\beta$ are defined in Corollary \ref{corollary_lemma_threshold_L_1}. Further, the Gaussian approximation on the error probability is obtained as
\begin{align}\label{Pe_approx_Gauss_M_2}
\nonumber
P_\textrm{e}\approx P_\textrm{e,G}&= \frac{1}{2}\big[1-Q(\frac{\beta-\sum_{l=1}^{L}w_l\lambda_{0,l}}{\sqrt{\sum_{l=1}^{L}w_l^2 \lambda_{0,l}}})\\
&\quad +Q(\frac{\beta-\sum_{l=1}^{L}w_l\lambda_{1,l}}{\sqrt{\sum_{l=1}^{L}w_l^2 \lambda_{1,l}}})\big],
\end{align}
where $Q(x)=\frac{1}{2\pi}\int_{x}^{\infty}\exp(-\frac{u^2}{2})du$. For $L=1$, the error probability reduces to
\small
\begin{equation}\label{Pe_L_1}
\begin{aligned}
P_\textrm{e}&=\frac{1}{2}\big[1-\sum_{y_1=0}^{\lfloor\mathcal{T}\rfloor}\frac{(\lambda_{0,1})^{y_1}\exp(-\lambda_{0,1})-(\lambda_{1,1})^{y_1}\exp(-\lambda_{1,1})}{y_1!}\big],
\end{aligned}
\end{equation}
\normalsize
for $\lambda_{0,1}>\lambda_{1,1}$, where $\mathcal{T}$ is defined in Corollary \ref{corollary_lemma_threshold_L_1}. 
Further, the Gaussian approximation on the error probability for $L=1$ is obtained as
\begin{align}\label{Pe_approx_Gauss_L_1}
P_\textrm{e}&=\frac{1}{2}\big[1+Q(\frac{\mathcal{T}-\lambda_{0,1}}{\sqrt{\lambda_{0,1}}})-Q(\frac{\mathcal{T}-\lambda_{1,1}}{\sqrt{\lambda_{1,1}}})\big],
\end{align}
\end{corollary1}
\begin{proof}
The proof is provided in Appendix \ref{AppendixProofcorollary3_1}.
\end{proof}

In the following, we obtain the Chernoff information (CI) upper bound on the error probability for the MAP detecter \cite{Chernoff1952, nielsen2011chernoff}.
Chernoff uses the inequality
\begin{align}\label{chernof_ineq}
\min(a,b)\leq a^s b^{1-s} \quad \forall s \in [0,1],
\end{align}
to upper bound the error probability of the MAP detector.
\begin{lemma}\label{lemma_chernoff_UPe_Lsps_1} The CI upper bound on the error probability with $M$ hypotheses 
and $L$-sample decoder 
is obtained as follows:
\begin{equation}\label{Peu_Lsps_1}
\begin{aligned}
&P_{\textrm{e}}\leq P_{\textrm{e,CI}}=\frac{M-1}{2}\max_{\substack{i_1,i_2 \in I, \\i_1\neq i_2} }
 \min_{s_{i_1,i_2}\in [0,1]} e^{-D_{i_1,i_2}(s_{i_1,i_2})},
\end{aligned}
\end{equation}
where $D_{i_1,i_2}(s_{i_1,i_2})=\sum_{l=1}^{L} [\lambda_{i_1,l}s_{i_1,i_2}+\lambda_{i_2,l}(1-s_{i_1,i_2})-\lambda_{i_1,l}^{s_{i_1,i_2}} \lambda_{i_2,l}^{1-s_{i_1,i_2}}]$. The optimum value of $s_{i_1,i_2}$ is the solution of the following equation:
\begin{equation}\label{opt_s_Lsps_1}
\begin{aligned}
\sum_{l=1}^{L}\big[\lambda_{i_1,l}-\lambda_{i_2,l}-\lambda_{i_1,l}(\frac{\lambda_{i_1,l}}{\lambda_{i_2,l}})^{s_{i_1,i_2}}\ln(\frac{\lambda_{i_1,l}}{\lambda_{i_2,l}})\big]=0.
\end{aligned}
\end{equation}
\end{lemma}
\begin{proof} The proof is provided in Appendix \ref{AppendixProoflemma3}.
\end{proof} 

There is no closed form solution for the optimum value of $s_{i_1,i_2}$ in \eqref{opt_s_Lsps_1}. In Corrollary \ref{corollary_chernoff_UPe_Lsps_1_1}, we use Holder's inequality, \cite{beckenbach1966holder}, to simplify the bound and obtain a closed form solution for the sub-optimum value of $s_{i_1,i_2}$.

\begin{corollary1}\label{corollary_chernoff_UPe_Lsps_1_1}
Using Holder's inequality on the CI upper bound, the error probability is upper bounded as follows:
\begin{equation}\label{Peu_Lsps_2}
\begin{aligned}
&P_{\textrm{e}}\leq P_{\textrm{e,HCI}}=\frac{M-1}{2}\max_{\substack{i_1,i_2 \in I,\\ i_1\neq i_2} }\min_{s_{i_1,i_2} \in [0,1]} e^{-K_{i_1,i_2}(s_{i_1,i_2})},
\end{aligned}
\end{equation}
where $K(s_{i_1,i_2})=(\sum_{l=1}^{L} \lambda_{i_1,l})s_{i_1,i_2}+(\sum_{l=1}^{L}\lambda_{i_2,l})(1-s_{i_1,i_2})-(\sum_{l=1}^{L}\lambda_{i_1,l})^{s_{i_1,i_2}} (\sum_{l=1}^{L}\lambda_{i_2,l})^{1-s_{i_1,i_2}}$. The optimum value of $s_{i_1,i_2}$ is obtained as
\begin{align}\label{opt_s_Lsps_2}
s_{i_1,i_2}^*=\frac{\ln(\frac{\sum_{l=1}^{L}\lambda_{i_1,l}}{\sum_{l=1}^{L}\lambda_{i_2,l}}-1)-\ln\ln(\frac{\sum_{l=1}^{L}\lambda_{i_1,l}}{\sum_{l=1}^{L}\lambda_{i_2,l}})}{\ln(\frac{\sum_{l=1}^{L}\lambda_{i_1,l}}{\sum_{l=1}^{L}\lambda_{i_2,l}})}.
\end{align}
\end{corollary1}
\begin{proof} 
The proof is provided in Appendix \ref{AppendixProofcorollary4_1}.
\end{proof}

\begin{corollary1}\label{corollary_chernoff_UPe_Lsps_1_2}
For binary hypothesis, the CI upper bound in \eqref{Peu_Lsps_1} reduces to
\begin{align}\label{Peu_Lsps_3}
P_{\textrm{e}}\leq P_{\textrm{e,u}}=\frac{1}{2} \min_{s \in (0,1)}\exp(-D(s)),
\end{align}
where $D(s)=\sum_{l=1}^{L}[\lambda_{0,l}s+\lambda_{1,l}(1-s)-\lambda_{0,l}^{s} \lambda_{1,l}^{1-s}]$. The optimum value of $s$ is the solution of the following equation:
\begin{align}\label{opt_s_Lsps_3}
\sum_{l=1}^{L}\big[\lambda_{0,l}-\lambda_{1,l}-\lambda_{0,l}(\frac{\lambda_{0,l}}{\lambda_{1,l}})^{s}\ln(\frac{\lambda_{0,l}}{\lambda_{1,l}})\big]=0.
\end{align}
For $L=1$, the CI upper bound is simplified as
\begin{align}\label{Peu_1sps_1}
&P_{\textrm{e}}\leq P_{\textrm{e,u}}=\frac{1}{2} e^{-(\lambda_{0,1}s^*+\lambda_{1,1}(1-s^*)-\lambda_{0,1}^{s^*} \lambda_{1,1}^{1-s^*})},
\end{align}
where 
\begin{align}\label{opt_s_1sps_1}
s^*=\frac{\ln(\frac{\lambda_{0,1}}{\lambda_{1,1}}-1)-\ln\ln(\frac{\lambda_{0,1}}{\lambda_{1,1}})}{\ln(\frac{\lambda_{0,1}}{\lambda_{1,1}})}.
\end{align}
\end{corollary1}

\textbf{Optimum and Sub-optimum sampling times}:
\label{subsec:sampling_times}
The optimum sampling times, which minimize the error probability, are
\begin{align}\label{tm_opt_1}
[t_{1}^*,t_{2}^*,...,t_{L}^*]=\arg \min_{t_{1},t_{2},...,t_{L}} P_\textrm{e},
\end{align}
where $P_\textrm{e}$ is given in \eqref{Pe_Poisson}. Since the above optimization problem is hard to solve in general case, the optimum sampling times should be obtained numerically. We use the Gaussian approximation and CI upper bound on the error probability and obtain the sub-optimum sampling times as the solutions of the following optimization problems:
\begin{align}\label{sub_opt_sampling_times_G}
&[t_{1,\textrm{G}},t_{2,\textrm{G}},...,t_{L,\textrm{G}}]=\arg \min_{t_{1},t_{2},...,t_{L}} P_\textrm{e,G},\\\label{sub_opt_sampling_times}
&[t_{1,\textrm{CI}},t_{2,\textrm{CI}},...,t_{L,\textrm{CI}}]=\arg \min_{t_{1},t_{2},...,t_{L}} P_\textrm{e,CI}=\\\nonumber
&\quad\quad\arg \max_{t_{1},t_{2},...,t_{L}} \min_{\substack{i_1,i_2 \in I, \\ i_1 \neq i_2}}\max_{s_{i_1,i_2}}D_{i_1,i_2}(s_{i_1,i_2}),
\end{align}
where $P_\textrm{e,G}$ is defined in \eqref{Pe_approx_Gauss_M_2} and $P_\textrm{e,CI}$, $D_{i_1,i_2}$ are defined in \eqref{Peu_Lsps_1}. 
In Lemma \ref{lemma_opt_tm_inf_L}, using the extension of Caratheodory's theorem \cite{Vershynin2017}, we obtain the sub-optimum sampling times using CI upper bound, when $L \rightarrow \infty$, in Lemma \ref{lemma_opt_tm_gauss}, we obtain the sub-optimum sampling times for binary hypothesis using Gaussian approximation of the error probability, and in Lemma \ref{lemma_opt_tm_chernoff}, we obtain the sub-optimum sampling times for binary hypothesis using CI upper bound.

\begin{lemma}\label{lemma_opt_tm_inf_L}
The sub-optimum sampling times using CI upper bound when $L \rightarrow \infty$ are ${M \choose 2}$ times, $t_{l},l=1,...,{M \choose 2}$, with weight $w_l$, i.e., $Lw_l$ sampling times are equal to $t_{l}$, where $t_{l}$s and $w_l$s are obtained from the following optimization problem:
\begin{equation}\label{optimization_problem_lemma5}
\begin{aligned}
\max_{w_1,...,w_{{M \choose 2}}} \max_{t_{1},...,t_{{M \choose 2}}} \min_{\substack{i_1,i_2 \in I, \\ i_1 \neq i_2}} \max_{s_{i_1,i_2}} \sum_{l=1}^{{M \choose 2}} w_l f_{i_1,i_2}(t_{l},s_{i_1,i_2}),
\end{aligned}
\end{equation}
\normalsize
where $f_{i_1,i_2}(t_{l},s_{i_1,i_2})=\lambda_{i_1,l}s_{i_1,i_2}+\lambda_{i_2,l}(1-s_{i_1,i_2})-\lambda_{i_1,l}^{s_{i_1,i_2}} \lambda_{i_2,l}^{1-s_{i_1,i_2}}$.
\end{lemma}
\begin{proof}
To obtain the optimum sampling times using CI upper bound in \eqref{Peu_Lsps_1}, we must solve
\begin{equation}\label{equation1}
\begin{aligned}
\max_{t_{1},...,t_{L}} \min_{\substack{i_1,i_2 \in I,\\ i_1 \neq i_2}} \max_{s_{i_1,i_2}} D_{i_1,i_2}(s_{i_1,i_2}),
\end{aligned}
\end{equation}
where $D_{i_1,i_2}(s_{i_1,i_2})=\sum_{l=1}^{L} [\lambda_{i_1,l}s_{i_1,i_2}+\lambda_{i_2,l}(1-s_{i_1,i_2})-\lambda_{i_1,l}^{s_{i_1,i_2}} \lambda_{i_2,l}^{1-s_{i_1,i_2}}]$. $\lambda_{i,l}$ is a function of the sampling time $t_{l}$. Hence, $D_{i_1,i_2}(s_{i_1,i_2})=\sum_{l=1}^{L}f_{i_1,i_2}(t_{l},s_{i_1,i_2})$. For each sampling time $t_{l}$, $l=1,...,L$, $\bm{a}_l=\big(f_{0,1}(t_{l},s_{0,1}),...,$ $f_{M-1,M}(t_{l},,s_{M-1,M})\big)$ is a point in $\mathbb{R}^{M \choose 2}$. The average of these points is
 $\frac{1}{L}\sum_{l=1}^{L}\bm{a}_l=\big(\frac{1}{L}\sum_{l=1}^{L}f_{0,1}(t_{l},s_{0,1}),...,\frac{1}{L}\sum_{l=1}^{L}f_{M-1,M}(t_{l},s_{M-1,M})\big).$
When $L \rightarrow \infty$, we have infinite points and the average point is in the convex hull of a set in $\mathbb{R}^{M \choose 2}$. Using the extension of Caratheodory's theorem for connected sets,\cite{Vershynin2017}, every point in the convex hull of a set $T$ in $\mathbb{R}^n$ can be expressed as a convex combination of at most $n$ points of $T$. Here, $n={M \choose 2}$, and we denote these ${M \choose 2}$ points by $\bm{b}_l=\big(f_{0,1}(t_{l}^{'},s_{0,1}),...,f_{M-1,M}(t_{l}^{'},,s_{M-1,M})\big)$, $l=1,..., {M \choose 2}$. Hence, $\frac{1}{L}\sum_{l=1}^{L}\bm{a}_l=\sum_{l=1}^{M \choose 2} $ $w_{l}\bm{b}_l$.
Thus, for a fixed $s_{i_1,i_2}$ and $i_1,i_2 \in I, i_1\neq i_2$, we have $\frac{1}{L}\sum_{l=1}^{L}f_{i_1,i_2}(t_{l},s_{i_1,i_2})=\sum_{l=1}^{M \choose 2} w_{l}f_{i_1,i_2}(t_{l}^{'},s_{i_1,i_2})$, and hence,
\begin{align}\label{eq_cara_1}
\nonumber
&\lim_{L \to \infty} \max_{t_{1},...,t_{L}} \min_{\substack{i_1,i_2 \in I,\\ i_1 \neq i_2}} \frac{1}{L} \sum_{l=1}^{L}f_{i_1,i_2}(t_{l},s_{i_1,i_2})=\\
&~ \max_{w_1,...,w_{{M \choose 2}}} \max_{t_{1}^{'},...,t_{{M \choose 2}}^{'}} \min_{\substack{i_1,i_2 \in I, \\ i_1 \neq i_2}} w_l f_{i_1,i_2}(t_{l}^{'},s_{i_1,i_2}).
\end{align}
Now, we are required to show the following expression:
\begin{subequations}\label{eq_twoside}
\begin{align}
&\lim_{L \to \infty} \max_{t_{1},...,t_{L}} \min_{\substack{i_1,i_2 \in I, \\ i_1 \neq i_2}} \max_{s_{i_1,i_2}} \frac{1}{L} \sum_{l=1}^{L}f_{i_1,i_2}(t_{l},s_{i_1,i_2}) =\label{eq_leftside}\\
&\max_{w_1,...,w_{{M \choose 2}}} \max_{t_{1}^{'},...,t_{{M \choose 2}}^{'}} \min_{\substack{i_1,i_2 \in I, \\i_1 \neq i_2}} \max_{s_{i_1,i_2}} \sum_{l=1}^{{M \choose 2}} w_l f_{i_1,i_2}(t_{l}^{'},s_{i_1,i_2}).\label{eq_rightside}
\end{align}
\end{subequations}
Let $s_{i_1,i_2}^*$ be the optimum value of $s_{i_1,i_2}$ for the optimization problem in \eqref{eq_leftside}. Then,
\begin{equation}
\begin{aligned}
&\lim_{L \to \infty} \max_{t_{1},...,t_{L}} \min_{\substack{i_1,i_2 \in I, \\i_1 \neq i_2}} \max_{s_{i_1,i_2}} \frac{1}{L} \sum_{l=1}^{L}f_{i_1,i_2}(t_{l},s_{i_1,i_2}) \\
&~ \geq \lim_{L \to \infty} \max_{t_{1},...,t_{L}} \min_{\substack{i_1,i_2 \in I, \\i_1 \neq i_2}} \frac{1}{L} \sum_{l=1}^{L}f_{i_1,i_2}(t_{l},s_{i_1,i_2}^*)\\
&~ \overset{(a)}= \max_{w_1,...,w_{{M \choose 2}}} \max_{t_{1}^{'},...,t_{{M \choose 2}}^{'}} \min_{\substack{i_1,i_2 \in I, \\i_1 \neq i_2}} w_l f_{i_1,i_2}(t_{l}^{'},s_{i_1,i_2}^*)\\\label{eq_1}
&~ = \max_{w_1,...,w_{{M \choose 2}}} \max_{t_{1}^{'},...,t_{{M \choose 2}}^{'}} \min_{\substack{i_1,i_2 \in I, \\ i_1 \neq i_2}} \max_{s_{i_1,i_2}} \sum_{l=1}^{{M \choose 2}}  w_l f_{i_1,i_2}(t_{l}^{'},s_{i_1,i_2}),
\end{aligned}
\end{equation}
where (a) follows from \eqref{eq_cara_1}. If $s_{i_1,i_2}^{**}$ is the optimum value of $s_{i_1,i_2}$  
for the optimization problem in \eqref{eq_rightside},
\begin{equation}
\begin{aligned}\label{eq_2}
&\max_{w_1,...,w_{{M \choose 2}}} \max_{t_{1}^{'},...,t_{{M \choose 2}}^{'}} \min_{\substack{i_1,i_2 \in I, \\ i_1 \neq i_2}} \max_{s_{i_1,i_2}} \sum_{l=1}^{{M \choose 2}} w_l f_{i_1,i_2}(t_{l}^{'},s_{i_1,i_2})\\
&\quad \geq \max_{w_1,...,w_{{M \choose 2}}} \max_{t_{1}^{'},...,t_{{M \choose 2}}^{'}} \min_{\substack{i_1,i_2 \in I, \\ i_1 \neq i_2}} \sum_{l=1}^{{M \choose 2}} w_l f_{i_1,i_2}(t_{l}^{'},s_{i_1,i_2}^{**})\\
&\quad \overset{(b)}=\lim_{L \to \infty} \max_{t_{1},...,t_{L}} \min_{\substack{i_1,i_2 \in I, \\ i_1 \neq i_2}} \frac{1}{L} \sum_{l=1}^{L}f_{i_1,i_2}(t_{l},s_{i_1,i_2}^{**})\\
&\quad =\lim_{L \to \infty} \max_{t_{1},...,t_{L}} \min_{\substack{i_1,i_2 \in I, \\ i_1 \neq i_2}} \max_{s_{i_1,i_2}} \frac{1}{L} \sum_{l=1}^{L}f_{i_1,i_2}(t_{l},s_{i_1,i_2}),
\end{aligned}
\end{equation}
\normalsize
where (b) follows from \eqref{eq_cara_1}. Hence, using \eqref{eq_1} and \eqref{eq_2}, we obtain \eqref{eq_twoside}.
\end{proof}

\begin{remark}
From the above lemma, for binary hypothesis, i.e., $M=2$, it can be easily seen that the sub-optimum sampling times when $L \rightarrow \infty$ are the same ($t_{1,\textrm{CI}}=...=t_{L,\textrm{CI}}$) and equal to the sampling time when $L=1$. This result is also true for the limited values of $L$, which is shown in Lemma \ref{lemma_opt_tm_chernoff}.\footnote{
Note that the problem of finding the optimum sampling times of the flow velocity detector can also be seen as either active hypothesis testing or channel discrimination problems, on which there are extensive literatures. In active hypothesis testing, the decision maker has control on the actions and the goal is to find the appropriate actions. 
In channel discrimination problem (with an extensive literature on quantum channels), there are a number of channels which we want to discriminate between and the inputs of the channels are chosen to have minimum error probability. The actions in active hypothesis testing and the inputs in channel discrimination problem are translated to sampling times in our model.
For $M=2$, in \cite{hayashi2009discrimination, nitinawarat2013controlled}, it is also shown that the actions that minimize CI upper bound are equal. 
}

\end{remark}

\begin{lemma}\label{lemma_opt_tm_gauss}
For binary hypothesis, the sub-optimum values of the sampling times $t_{1},t_{2},...,t_{L}$ using Gaussian approximation are the solutions of:
\begin{equation}\label{setofeqs_tms_opt_2}
\begin{aligned}
& \frac{1}{\sigma_0}e^{\frac{-(\beta-\mu_0)^2}{2\sigma_0^2}}\big[g_{0,l}-g_{1,l}-(\frac{g_{0,l}}{\lambda_{0,l}}-\frac{g_{1,l}}{\lambda_{1,l}})\lambda_{0,l}(1+\frac{1}{\sigma_0^2})\\
&~ -w_l g_{0,l}(1+\frac{1}{2\sigma_0^2})\big]-\frac{1}{\sigma_1}e^{\frac{-(\beta-\mu_1)^2}{2\sigma_1^2}} \big[g_{0,l}-g_{1,l}\\
&~ -(\frac{g_{0,l}}{\lambda_{0,l}}-\frac{g_{1,l}}{\lambda_{1,l}})\lambda_{1,l}(1+\frac{1}{\sigma_1^2})-w_l g_{1,l}(1+\frac{1}{2\sigma_1^2})\big]=0,
\end{aligned}
\end{equation}
for $l=1,...,L$, where $\mu_i=\sum_{l=1}^{L}w_l\lambda_{i,l}$, $\sigma_i=\sqrt{\sum_{l=1}^{L}w_l^2\lambda_{i,l}}$, $g_{i,l}=\frac{d}{d t_{l}}\lambda_{i,l}$, $i=0,1$, and $\beta$ and $w_l$ are defined in Corollary \ref{corollary_lemma_threshold_L_1}. For the location invariant flow velocity and the transparent receiver, we have $g_{i,l}=\lambda_{i,l}\big[\frac{-3}{2(t_{l}-t_\textrm{r})}+\frac{\langle \bm{v}_i(t_l),\bm{r}_0-\int_{t_\textrm{r}}^{t_{l}}\bm{v}_i(\tau)d\tau \rangle }{2D(t_l-t_\textrm{r})}+\frac{||\bm{r}_0-\int_{t_\textrm{r}}^{t_{l}}\bm{v}_i(\tau)d\tau||^2}{4D(t_l-t_\textrm{r})^2} \big]$, $i=0,1$.
\end{lemma}
\begin{proof}
The proof is provided in Appendix \ref{AppendixProoflemma6}.
\end{proof}
\begin{lemma}\label{lemma_opt_tm_chernoff}
The sub-optimum values of the sampling times for binary hypothesis and $L$ sample receiver using CI upper bound are the same and equal to the sampling time of $L=1$, noted by, $t_1$, which is the solution of:
\begin{align}\label{setofeqs_tms_opt_1}
\begin{cases}
g_{0,1}s+g_{1,1}(1-s)-s g_{0,1}(\frac{\lambda_{1,1}}{\lambda_{0,1}})^{1-s}\\
\qquad\qquad\qquad\qquad -(1-s) g_{1,1} (\frac{\lambda_{0,1}}{\lambda_{1,1}})^{s}=0,\\
s=\frac{\ln(\frac{\lambda_{0,1}}{\lambda_{1,1}}-1)-\ln\ln(\frac{\lambda_{0,1}}{\lambda_{1,1}})}{\ln(\frac{\lambda_{0,1}}{\lambda_{1,1}})},
\end{cases}
\end{align}
where $g_{i,1}$ is defined in Lemma \ref{lemma_opt_tm_gauss}.
\end{lemma}
\begin{proof}
The proof is provided in Appendix \ref{AppendixProoflemma7}.
\end{proof}

\section{Flow velocity estimator} \label{sec3:flow_estimator}
Here, we obtain the estimation of the flow velocity for the $L$-sample receiver with independent observations $y_1,...,y_L$ in time instances $t_1,...,t_L$. We denote the mean number of the received molecules in sampling time $t_l$ for the flow velocity $\bm{v}$ as $\lambda_{l}(\bm{v}), l=1,...,L$.
For the transparent receiver, $\lambda_{l}(\bm{v})=V_\textrm{R}\zeta h_0(\bm{r}_0- \bm{v}.(t_l-t_\textrm{r}),t_l)$. We obtain the MAP and MMSE estimators for a randomly chosen constant flow velocity, i.e., $\bm{v}(\bm{r},t)=\bm{v}=(v_x,v_y,v_z)$.
We assume that $v_x$, $v_y$, and $v_z$ are independent with prior probability distribution functions (PDF) as $p_x(v_x)$, $p_y(v_y)$, and $p_z(v_z)$, respectively. To investigate the error performance of the estimators, we consider the minimum mean square error (MSE) of the estimators and obtain the Bayesian Cramer-Rao (BCR) lower bound on the MSE of the estimators.

\begin{lemma}\label{lemma_MAP_est_L}
For a molecular flow velocity estimator to estimate randomly chosen constant flow velocity, the MAP estimator is obtained as
\begin{equation}\label{estimation_rule_csk}
\begin{aligned}
\nonumber
\hat{\bm{v}}=\arg\max_{\bm{v}}\sum_{l=1}^{L} [y_l \ln (\lambda_{l}(\bm{v}))-\lambda_{l}(\bm{v})]+\ln(p_{x,y,z}(\bm{v})),
\end{aligned}
\end{equation}
where $p_{x,y,z}(\bm{v})=p_x(v_x)p_y(v_y)p_z(v_z)$. Hence, if $p_{x,y,z}(\bm{v})$ and $\lambda_{l}(\bm{v})$ are differentiable with respect to $\bm{v}$, $\hat{v}_x$, $\hat{v}_y$, and $\hat{v}_z$ are the solutions of the following set of equations:
\begin{equation}\label{estimation_rule_2}
\begin{aligned}
\sum_{l=1}^{L} \frac{1}{\lambda_{l}(\bm{v})}.\frac{\partial \lambda_{l}(\bm{v})}{\partial v_i}.(y_l-\lambda_l(\bm{v}))+\frac{1}{p_i(v_i)}.\frac{d p_i(v_i)}{d v_i}=0.
\end{aligned}
\end{equation}
For the transparent receiver, we have $\frac{\partial \lambda_{l}(\bm{v})}{\partial v_i}=\frac{(\bm{r}_{0,i}-v_i.(t_l-t_\textrm{r}))}{2D}\lambda_l(\bm{v})$, $ i \in \{x,y,z\}$.
\end{lemma}
\begin{proof} 
The proof is straightforward similar to the proof of MAP detector in Lemma \ref{lemma_threshold_L}.
\end{proof}

\begin{corollary1} \label{corollary_MAP_est_L}
For one-sample receiver and the location invariant flow velocity in the direction of the connecting line between the releasing node and the transparent receiver, i.e, $\bm{v}=v \bm{d}$, where $\bm{d}=\frac{\bm{r}_0}{r_0}$, with uniform priori PDF for $v$ in the range $S_v=[v_\textrm{min},v_\textrm{max}]$, the MAP estimator of $v$ is obtained as
\begin{equation}\label{estimation_rule_csk_map}
\begin{aligned}
\hat{v}=\begin{cases}
v_1, \quad &\textrm{if}~y_1 \geq \lambda_1(v_1 \bm{d}), v_1 \in S_v,\\
v_2, \quad &\textrm{if}~0 < y_1 < \lambda_1(v_1 \bm{d}), v_2 \in S_v,\\
v_3, \quad &\textrm{if}~0 < y_1 < \lambda_1(v_1 \bm{d}), v_3 \in S_v,\\
v_{\textrm{min}},\quad &\textrm{if}~A_1\cap B_1,\\
v_{\textrm{max}},\quad &\textrm{if}~A_2\cap B_2,
\end{cases}
\end{aligned}
\end{equation}
where $v_1=\frac{r_0}{t_1-t_\textrm{r}}$, $v_2=\frac{r_0+\sqrt{-4D (t_1-t_\textrm{r})(\ln{y_l}-\ln{(\lambda_{1}(v_1\bm{d}))})}}{t_l-t_\textrm{r}}$, $v_3=\frac{r_0-\sqrt{-4D (t_1-t_\textrm{r})(\ln{y_l}-\ln{(\lambda_{1}(v_1\bm{d}))})}}{t_l-t_\textrm{r}}$, $A_1=\{y_1 \geq \lambda_1(v_1 \bm{d}), v_1 \notin S_v\} \cup \{0 <y_1 < \lambda_1(v_1 \bm{d}), v_2 \notin S_v, v_3 \notin S_v\} \cup \{y_1=0 \big\}$, $A_2=\{y_1 \geq \lambda_1(v_1 \bm{d}), v_1 \notin S_v\} \cup \{0 <y_1 < \lambda_1(v_1 \bm{d}), v_2 \notin S_v, v_3 \notin S_v\}\cup \{y_1=0 \big\}$,
$B_1=\{R_\textrm{est,u}(v_\textrm{min}) \geq R_\textrm{est,u}(v_\textrm{max})\}$, and $B_2=\{R_\textrm{est,u}(v_\textrm{min})\leq R_\textrm{est,u}(v_\textrm{max}\}$, in which $R_\textrm{est,u}(v)=y_1 \ln (\lambda_{1}(v\bm{d}))-\lambda_{1}(v\bm{d})$. Note that when the estimator gives two values, one of them is chosen randomly as the estimated value.

\begin{proof}The proof is provided in Appendix \ref{AppendixProofcorrollary8_1}.
\end{proof}
\end{corollary1}

\begin{corollary1}
If the sampling times are equal, we have $\lambda_{1}(\bm{v})=...=\lambda_L(\bm{v})$. Hence, from \eqref{estimation_rule_2}, we should find the solutions of 
\begin{equation}\label{eq_equal_sampling_times}
\begin{aligned}
\frac{1}{\lambda_{l}(\bm{v})}\frac{\partial \lambda_{l}(\bm{v})}{\partial v_i}.(\sum_{l=1}^{L}y_l-L\lambda_l(\bm{v}))+\frac{1}{p_i(v_i)}\frac{d p_i(v_i)}{d v_i}=0,
\end{aligned}
\end{equation}
for $ i \in \{x,y,z\}$, to obtain the estimated values of $v_x,v_y$, and $v_z$. For the transparent receiver and the flow velocity in the direction of the releasing node and the receiver with uniform prior pdf for its magnitude, we should find the solution of $(r_{0}-v.(t_1-t_\textrm{r})).(\frac{1}{L}\sum_{l=1}^{L}y_l-\lambda_{1}(v\bm{d}))=0$. Hence, the procedure to obtain the estimated value of $v$ is similar to the one-sample receiver which is obtained in Corollary \ref{corollary_MAP_est_L}, with the difference that we should use $\frac{1}{L}\sum_{l=1}^{L}y_l$ instead of $y_l$ in the equations, i.e., we should take the average of the samples and replace it as the observation value in the one-sample receiver.
\end{corollary1}

\begin{lemma}\label{lemma_MMSE_est_L}
The MMSE estimator to estimate randomly chosen constant flow velocity with finite mean and variance is obtained as
\begin{equation}\label{estimation_rule_MMSE}
\begin{aligned}
\hat{v}_i&=\E[v_i|y_1,...,y_L]=\frac{\int v_i \exp(R_{\textrm{est}}(\bm{v})) dv_z dv_y dv_x}{\int \exp(R_{\textrm{est}}(\bm{v}))  dv_z dv_y dv_x },
\end{aligned}
\end{equation}
for $i \in \{x,y,z\}$, where $R_{\textrm{est}}(\bm{v})=\sum_{l=1}^{L} [y_l \ln (\lambda_{l}(\bm{v}))-\lambda_{l}(\bm{v})]+\ln(p_{x,y,z}(\bm{v}))$. The linear MMSE (LMMSE) estimator is obtained as:
\begin{equation}\label{estimation_rule_csk_MMSE}
\begin{aligned}
\hat{v}_i=\sum_{l=1}^{L}\frac{\C(Y_l,v_i)}{\V(Y_l)}(y_l-\E[Y_l])+\E[v_i],\quad i \in \{x,y,z\},
\end{aligned}
\end{equation}
where for $l=1,...,L$ and $i\in\{x,y,z\}$,
\begin{align}
\E[Y_l]&=\int p_{x,y,z}(\bm{v})\lambda_l(\bm{v}) dv_z dv_y dv_x,\\\nonumber
\E[Y_l^2]&=\int p_{x,y,z}(\bm{v})\lambda_l(\bm{v})(1+\lambda_l(\bm{v}))  dv_z dv_y dv_x,\\\nonumber
\C(Y_l,v_i)
&=\int p_{x,y,z}(\bm{v}) (v_i-\E[v_i])\lambda_l(\bm{v}) dv_z dv_y dv_x.
\end{align}
\end{lemma}
\begin{proof}
The proof is provided in Appendix \ref{AppendixProoflemma9}.
\end{proof}

\begin{corollary1}\label{corrolary_LMMSE}
For $\bm{v}=v\bm{d}$, with uniform prior PDF for $v$ in range $S_v=[v_\textrm{min},v_\textrm{max}]$, the MMSE estimator is obtained as
\begin{align}
&\hat{v}=\frac{\int_{S_v}  v \exp\big(\sum_{l=1}^{L}[y_l \ln(\lambda_l(v\bm{d}))-\lambda_l(v\bm{d})]\big) dv}{\int_{S_v} \exp\big(\sum_{l=1}^{L}[y_l \ln(\lambda_l(v\bm{d}))-\lambda_l(v\bm{d})]\big) dv},
\end{align}
and the LMMSE estimator is obtained as
\begin{align}
&\hat{v}=\sum_{l=1}^{L}\frac{\C(Y_l,v)}{\V(Y_l)}(y_l-\E[Y_l])+\frac{v_{+}}{2},
\end{align}
where for $l=1,...,L$,
\begin{align}
&\E[Y_l]=\frac{1}{v_{-}}\int_{S_v} \lambda_l(v\bm{d}) dv, \\\nonumber
& \E[Y_l^2]=\frac{1}{v_{-}}\int_{S_v}\lambda_l(v\bm{d})(1+\lambda_l(v\bm{d}))  dv,\\\nonumber
&\C(Y_l,v)=\frac{1}{v_{-}}\int_{S_v} (v-\frac{v_{+}}{2})\lambda_l(v\bm{d}) dv,
\end{align}
where $v_{+}=v_\textrm{min}+v_\textrm{max}$ and $v_{-}=v_\textrm{max}-v_\textrm{min}$.
\end{corollary1}


\begin{corollary1}
If the sampling times are equal, the LMMSE estimator of $v_i$, $i \in \{x,y,z\}$, is obtained as 
\begin{align}
\hat{v}_i=L\frac{\C(v_iY_1)}{\V(Y_1)}(\frac{1}{L}\sum_{l=1}^{L}y_l-\E[Y_1])+\E[v_i].
\end{align}
\end{corollary1}

In the following, we investigate the error performance of the estimators. The estimation error is $\bm{\epsilon}=\bm{v}-\hat{\bm{v}}$, where $\epsilon_i$ is a random variable ($v_i$s are random variables with prior PDF $p_i(v_i)$ and $\hat{v}_i$ is a function of Poisson random variables $Y_1,..,Y_L$). To investigate the performance of the estimators, we consider the MSE of the estimation, i.e., $\E[\epsilon^2]$ (where $\epsilon=||\bm{\epsilon}||_2$), which is hard to compute in general case. In Section \ref{Simulation}, we obtain the MSE of the considered estimators numerically. In the following, we obtain the Bayesian and expected Cramer-Rao lower bounds on the MSE.
The Bayesian Cramer-Rao (BCR) lower bound on the MSE is defined as \cite{van2004detection}
\begin{equation}\label{cramer_rao_bound_BCR}
\begin{aligned}
\E[\epsilon^2]\geq \textrm{Tr}\{J_B^{-1}\}, \qquad J_B=-\E_{\bm{v},\bm{Y}}[\nabla^2_{\bm{v}\bm{v}}(\ln(\Prob(\bm{v},\bm{y}))],
\end{aligned}
\end{equation}
where $\bm{Y}=(Y_1,...,Y_L)$. $J_B$ can be divided into two matrixes $J_P$ and $J_D$:
\begin{align}\label{cramer_rao_bound2}
\nonumber
&J_B=J_D+J_P, \qquad J_D=\E_{\bm{v}}[J_F(\bm{v})], \\
&J_P=-\E_{\bm{v}}[\nabla^2_{\bm{v}\bm{v}}(\ln(\Prob(\bm{v}))],
\end{align}
where 
\begin{align}\label{cramer_rao_bound3}
J_F(\bm{v})=-\E_{\bm{Y}|\bm{v}}[\nabla^2_{\bm{v}\bm{v}}(\ln(\Prob(\bm{Y}|\bm{v}))]
\end{align}
is Fisher's information matrix. The following conditions must hold for the BCR lower bound:
\begin{itemize}
\item{$\frac{\partial \ln(\Prob(\bm{v},\bm{Y}))}{\partial v_i}$ and $\frac{\partial^2 \ln(\Prob(\bm{v},\bm{Y}))}{\partial v_i \partial v_j}$, for $i,j \in \{x,y,z\}$, are absolutely integrable with respect to $\bm{v}$ and $\bm{Y}$.}
\item{$\lim_{v_i \rightarrow \pm \infty}b(\bm{v})\Prob(\bm{v})=0$, for $i \in \{x,y,z\}$, where $\bm{b}(\bm{v})$ is called the bias function defined as
\begin{align}\label{biasfunction}
\bm{b}(\bm{v})=E_{\bm{Y}|\bm{v}}[\hat{\bm{v}}]-\bm{v}.
\end{align}
}
\end{itemize}
The BCR lower bound is obtained in the following lemma. 
\begin{lemma}\label{lemma_CR_est_L_performance} 
The BCR lower bound on the MSE to estimate randomly chosen constant flow velocity is obtained as
\begin{align}\label{CR_lower_bound}
&\E[\epsilon^2]\geq \textrm{Tr}\{(J_D+J_P)^{-1}\},\\\nonumber
& \{J_D\}_{i,j}=\sum_{l=1}^{L}\E_{\bm{v}}\big[\frac{1}{\lambda_l(\bm{v})}.\frac{\partial \lambda_l(\bm{v})}{\partial v_i}. \frac{\partial \lambda_l(\bm{v})}{\partial v_j}\big],\\\label{CR_lower_bound2}
& \{J_P\}_{i,j}=\begin{cases}
-\sum_{l=1}^{L}\E_{v_i}[\frac{d^2 \ln (p_i(v_i)}{d v_i^2})\big)], &i=j\\
0, &i \neq j
\end{cases},
 \end{align}
for $i,j \in \{x,y,z\}$, where the following conditions must hold:
\begin{itemize}
\item{$\frac{\partial \ln(\Prob(\bm{v},\bm{Y}))}{\partial v_i}$ and $\frac{\partial^2 \ln(\Prob(\bm{v},\bm{Y}))}{\partial v_i \partial v_j}$ are absolutely integrable with respect to $\bm{v}$ and $\bm{Y}$.}
\item{$\lim_{v_i \rightarrow \pm \infty}b(\bm{v})p_{x,y,z}(\bm{v})=0$, where $p_{x,y,z}(\bm{v})$ is defined in Lemma \ref{lemma_MAP_est_L} and $b(\bm{v})$ is defined in \eqref{biasfunction}.}
\end{itemize}
\end{lemma}
\begin{proof}The proof is provided in Appendix \ref{AppendixProoflemma10}.
\end{proof}
Although the BCR lower bound is valid for both biased and unbiased estimators, due to the conditions which must hold for the BCR lower bound, this lower bound is not valid when prior distribution is bounded (e.g., uniform distribution). Another lower bound on the MSE is the expected Cramer-Rao (ECR) lower bound, defined as \cite{ben2009lower}
\small
\begin{equation}\label{cramer_rao_bound_ECR1}
\begin{aligned}
\E[\epsilon^2]&\geq \E_{\bm{v}}[ \textrm{Tr}\{(1+\bm{b}^{'}(\bm{v}))J_F^{-1}(\bm{v})(1+\bm{b}^{'}(\bm{v}))^T+||\bm{b}(\bm{v})||^2 \}],
\end{aligned}
\end{equation}
\normalsize
where, $J_F(v)$ and $\bm{b}(\bm{v})$ are defined in \eqref{cramer_rao_bound3}. The following condition must hold for the ECR lower bound:
\begin{itemize}
\item{$\frac{\partial \ln(\Prob(\bm{Y}|\bm{v}))}{\partial v_i}$ and $\frac{\partial^2 \ln(\Prob(\bm{Y}|\bm{v}))}{\partial v_i \partial v_j}$, for $i,j \in \{x,y,z\}$, are absolutely integrable.}
\end{itemize}
Since there is no condition on the distribution of $\bm{v}$, the ECR lower bound is valid for all distributions over $v$ including the bounded distributions. 
However, the bias function $\bm{b}$ should be obtained for each estimator, which might be challenging. In \cite{ben2009lower}, $\bm{b}$ is optimized to obtain a general lower bound on all estimators. For $\bm{v}=v\bm{d}$, and bounded distributions for $v$, i.e, $v \in [v_\textrm{min},v_\textrm{max}]$, the optimal bias function $b(v)$ is the solution of the following differential equation \cite{ben2009lower}:
\small
\begin{equation}\label{cramer_rao_bound_ECR2}
\begin{aligned}
J_F(v)b(v)&=b^{''}(v)+(1+b^{'}(v))\big(\frac{d \ln P(v)}{dv}-\frac{d \ln(J_F(v))}{dv}\big),
\end{aligned}
\end{equation}
\normalsize
within the range $v \in [v_\textrm{min},v_\textrm{max}]$, with boundary condition $b^{'}(v_\textrm{min})=b^{'}(v_\textrm{max})=-1$. The ECR lower bound on the MSE is obtained in the following lemma for the transparent receiver and $\bm{v}=v\bm{d}$ with uniform distribution for $v \in [v_\textrm{min}, v_\textrm{max}]$. 

\begin{lemma}\label{lemma_ECR_est_L_performance} 
For $L=1$ and $\bm{v}=v\bm{d}$, with uniform prior PDF for $v$ in range $S_v=[v_\textrm{min},v_\textrm{max}]$, the ECR lower bound is obtained as
\begin{align}\label{ECR_lower_bound}
\E[\epsilon^2]\geq \frac{1}{v_{-}}\int_{S_v} \big[\frac{(1+b^{'}(v))^2}{J_F(v)}+b^2(v)\big] dv,
 \end{align}
where $J_F(v)=\frac{1}{4D^2}(r_0-v.(t_l-t_\textrm{r}))^2 \lambda_1(v\bm{d})$, $v_{-}$ is defined in Corrollary \ref{corrolary_LMMSE}, and the optimal bias function $b(v)$ is the solution of the following ordinary differential equation:
\begin{align}\label{ECR_lower_bound}
\nonumber
&J_F(v)b(v)=b^{''}(v)-(1+b^{'}(v))\big(\frac{r_0-v.(t_l-t_\textrm{r})}{2D}\\
&\qquad \qquad -\frac{2(t_l-t_\textrm{r})}{r_0-v.(t_l-t_\textrm{r})}\big), \quad v \in (v_\textrm{min},v_\textrm{max}),
\end{align}
with condition $b^{'}(v_\textrm{min})=b^{'}(v_\textrm{max})=-1$.
\end{lemma}
\begin{proof}
The proof is straightforward from \eqref{cramer_rao_bound_ECR1}, \eqref{cramer_rao_bound_ECR2}.
\end{proof}

\textbf{Optimum and sub-optimum sampling times}:
Te obtain the optimum sampling times, we minimize the MSE, for $t_1,...,t_L$, i.e., 
\begin{align}\label{opt_sampling_times}
[t_1^*,t_2^*...t_L^*]=\arg \min_{t_1,t_2,...,t_L} \E[\epsilon^2].
\end{align}
However, 
the distribution of $\bm{\epsilon}$ is hard to compute in general case. 
In Section \ref{Simulation}, we obtain the estimation error and the optimum sampling times numerically.
In Lemma \ref{lemma_opt_sampling_times_infL}, when $L \rightarrow \infty$, we obtain the optimum sampling times for an MAP estimator of the magnitude of the flow velocity which is in the direction of the connecting line between the releasing node and the receiver with uniform distribution. 
\begin{lemma} 
\label{lemma_opt_sampling_times_infL}
The optimum sampling times for an MAP estimator of $\bm{v}=v\bm{d}$, with uniform distribution for $v$, when $L \rightarrow \infty$ are at most two distinct times $t_1$ and $t_2$ with weights $\tilde{w}_1$ and $\tilde{w}_2$, respectively, i.e., $L{\tilde w_1}$ sampling times are equal to $t_1$ and $L{\tilde w_2}$ sampling times are equal to $t_2$. The two sampling times and their weights can be obtained from \eqref{opt_sampling_times} numerically.
\end{lemma}
\begin{proof}When $L \rightarrow \infty$, if the magnitude of the flow velocity is $v_r$, the average value of the observations $\frac{1}{L}\sum_{l=1}^{L}y_l$ approaches to $\lambda_{1}(v_r\bm{d})$. Hence, if the samples are taken at the same time, from \eqref{eq_equal_sampling_times}, $\hat{v}$ is the solution of $\frac{1}{L}\sum_{l=1}^{L}y_l-\lambda_{1}(v\bm{d})=0$, we obtain $v_r$ and $v_r+\frac{2 r_0}{(t_1-t_\textrm{r})}$ as the maximizers of $R_\textrm{est}(v\bm{d})$, which means that we may have ambiguity on the estimated flow velocity. 
This is because of the fact that the function $\lambda_l(v\bm{d})$ is not a one by one function of $v$. Since logarithm of the function $\lambda_l(v\bm{d})$ is a Quadratic function, using two values of $\lambda_l(v\bm{d})$, for two different sampling times, we can obtain $v$, i.e., if we have $\lambda_{1}(v\bm{d})=a_1$ and $\lambda_{2}(v\bm{d})=a_2$, we can obtain $v$ uniquely. Note that the flow velocity is chosen randomly and for each flow velocity, every two different sampling times leads to a unique estimation. Hence, we conclude that when $L \rightarrow \infty$, if we have two different times, the estimation error approaches to zero. 
\end{proof}

\renewcommand{\arraystretch}{0.7}
\begin{figure}
\centering
\captionof{table}{Simulation and numerical analysis parameters}
\begin{tabular}{l|l}
\hline\hline
Parameter & Value\\\hline
$D$ & $10^{-8}~\textrm{m}^2/\textrm{s}$\\
$\zeta$ & $10000$\\
$r_0$ & $100 ~\mu \textrm{m}$\\
$r_\textrm{R}$ & $1.5 \times 10^{-5}$\\
$t_\textrm{r}$ & $0$\\
\hline\hline
\end{tabular}
\label{parameters}
\vspace{-1.5em}
\end{figure}

\vspace{-0.2em}
\section{Simulation and Numerical Results}\label{Simulation}
In this section, we provide some simulation and numerical results to evaluate the performance of the proposed flow velocity detector and estimator. For the evaluation, we use the parameters given in table \ref{parameters}. In part A, we consider the flow velocity detector, and in part B, we consider the flow velocity estimator. In both parts, we assume that the flow velocity is in the direction of the connecting line between the releasing node and the receiver.

\subsection{Flow velocity detector}\label{simulation_ook}
Here, we assume binary and multiple hypotheses for the flow velocity. We assume that the hypotheses in the flow velocity detector are location and time invariant, i.e., $\bm{v}_i(\bm{r},t)=v_i\bm{d}, i \in I$.

\begin{figure*}
\centering
\begin{subfigure}[t]{0.23\textwidth}
\centering
\includegraphics[trim={1.5cm 0 0 0},scale=0.25]{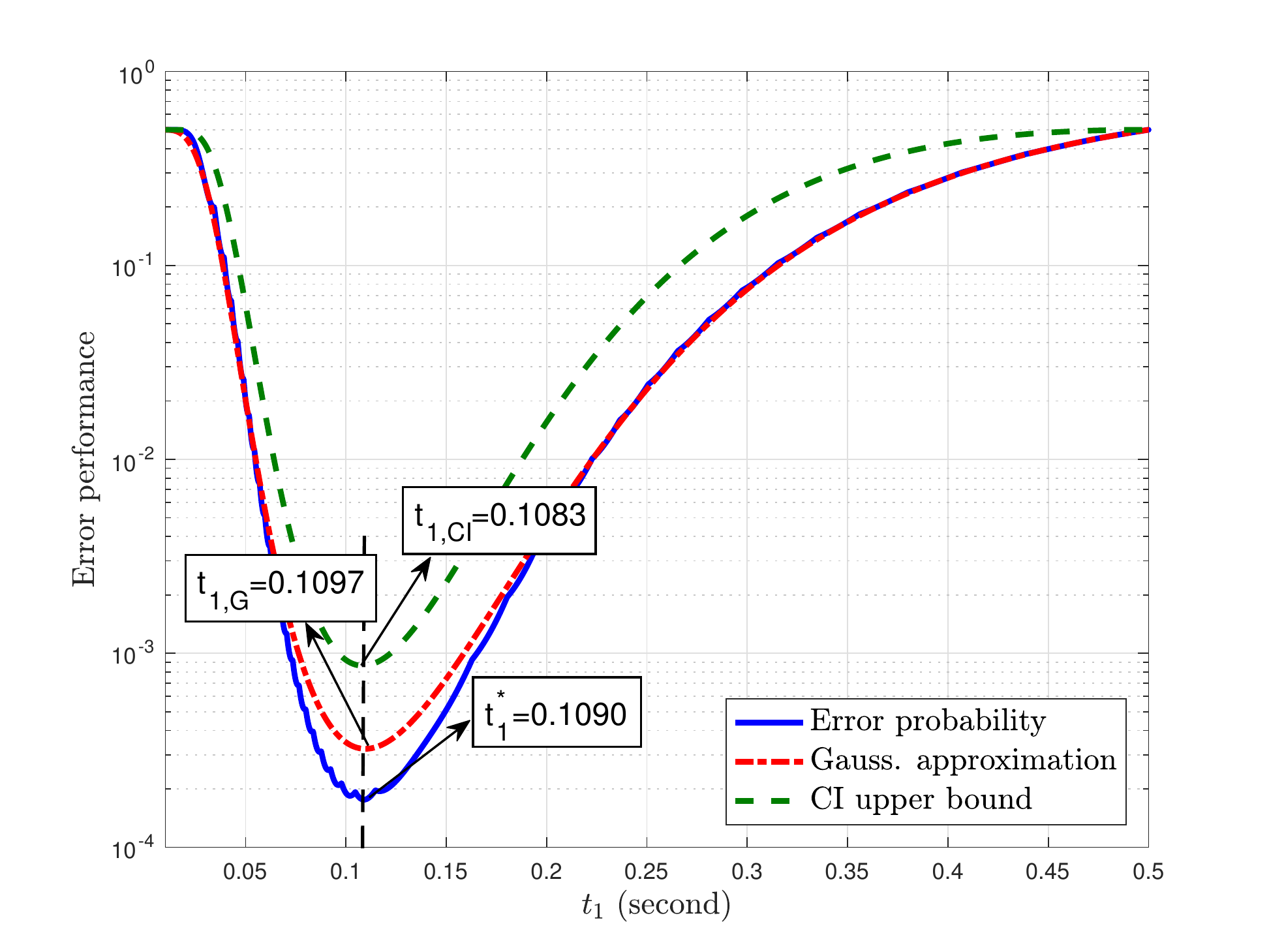}
\vspace{-1.5em}
\caption{Error performance versus $t_1$ for $L=1$.}
\label{fig_pe_tm}
\end{subfigure}\quad
\begin{subfigure}[t]{0.23\textwidth}
\centering
\includegraphics[trim={1.5cm 0 0 0},scale=0.25]{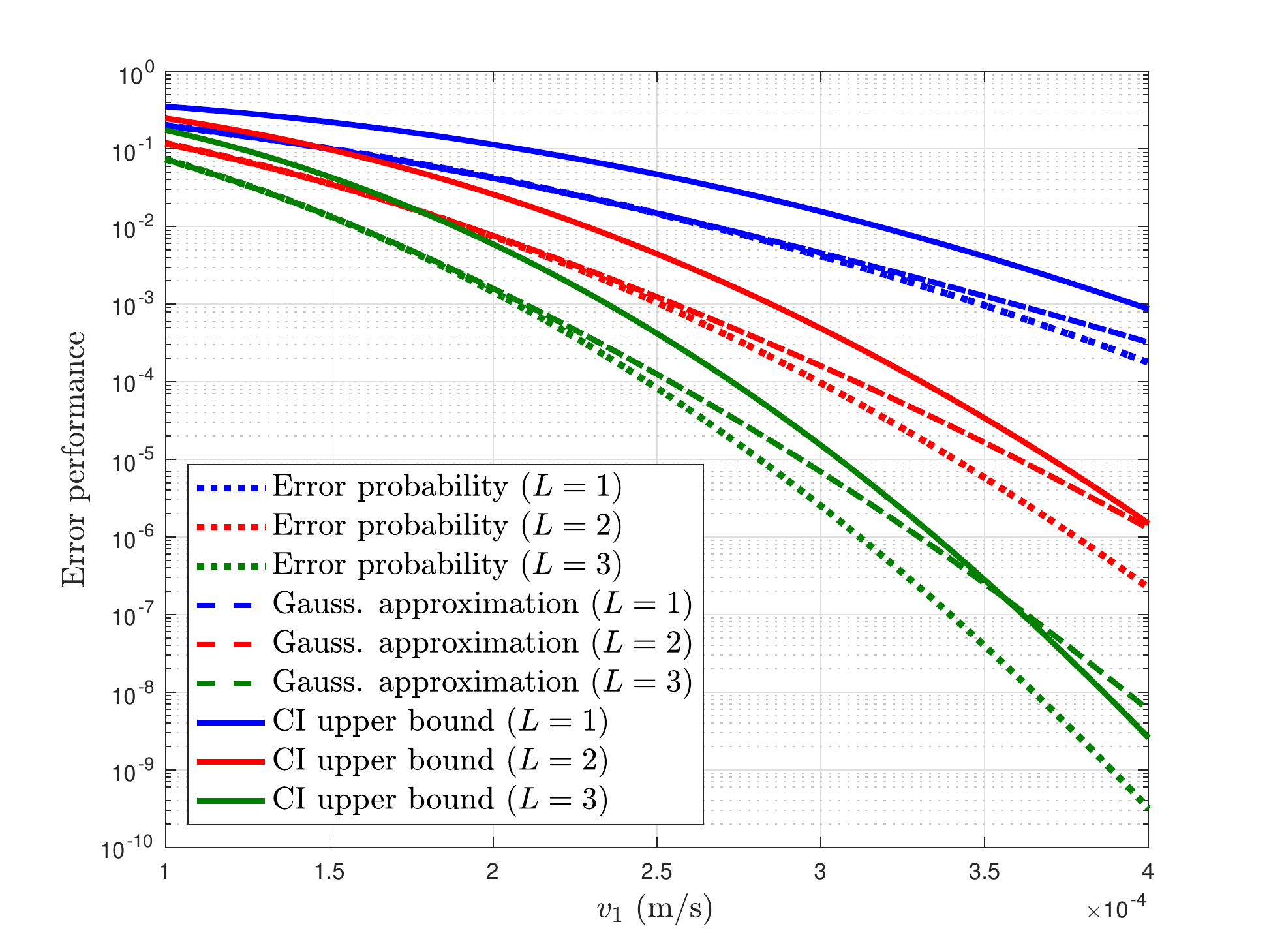}
\vspace{-1.5em}
\caption{Error performance versus $v_1$ for $L=1,2,3$ with related optimum sampling times.}
\label{fig_pe_v0}
\end{subfigure}\quad
\begin{subfigure}[t]{0.23\textwidth}
\centering
\includegraphics[trim={1.5cm 0 0 0cm},scale=0.25]{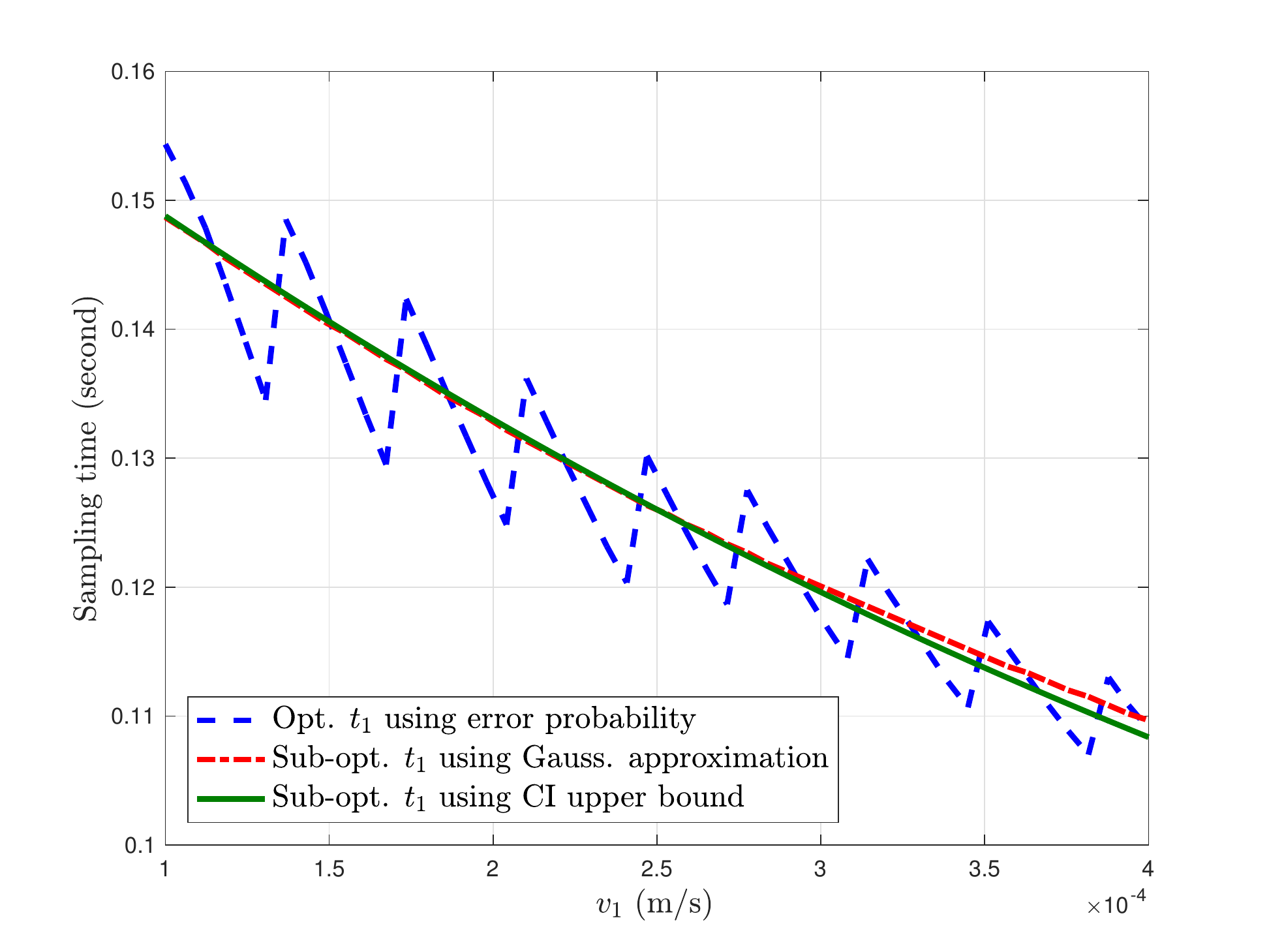}
\vspace{-1.5em}
\caption{Sampling times versus $v_1$ for $L=1$.}
\label{fig_tm}
\end{subfigure}\quad 
\begin{subfigure}[t]{0.23\textwidth}
\centering
\includegraphics[trim={1.5cm 0 0 0cm},scale=0.25]{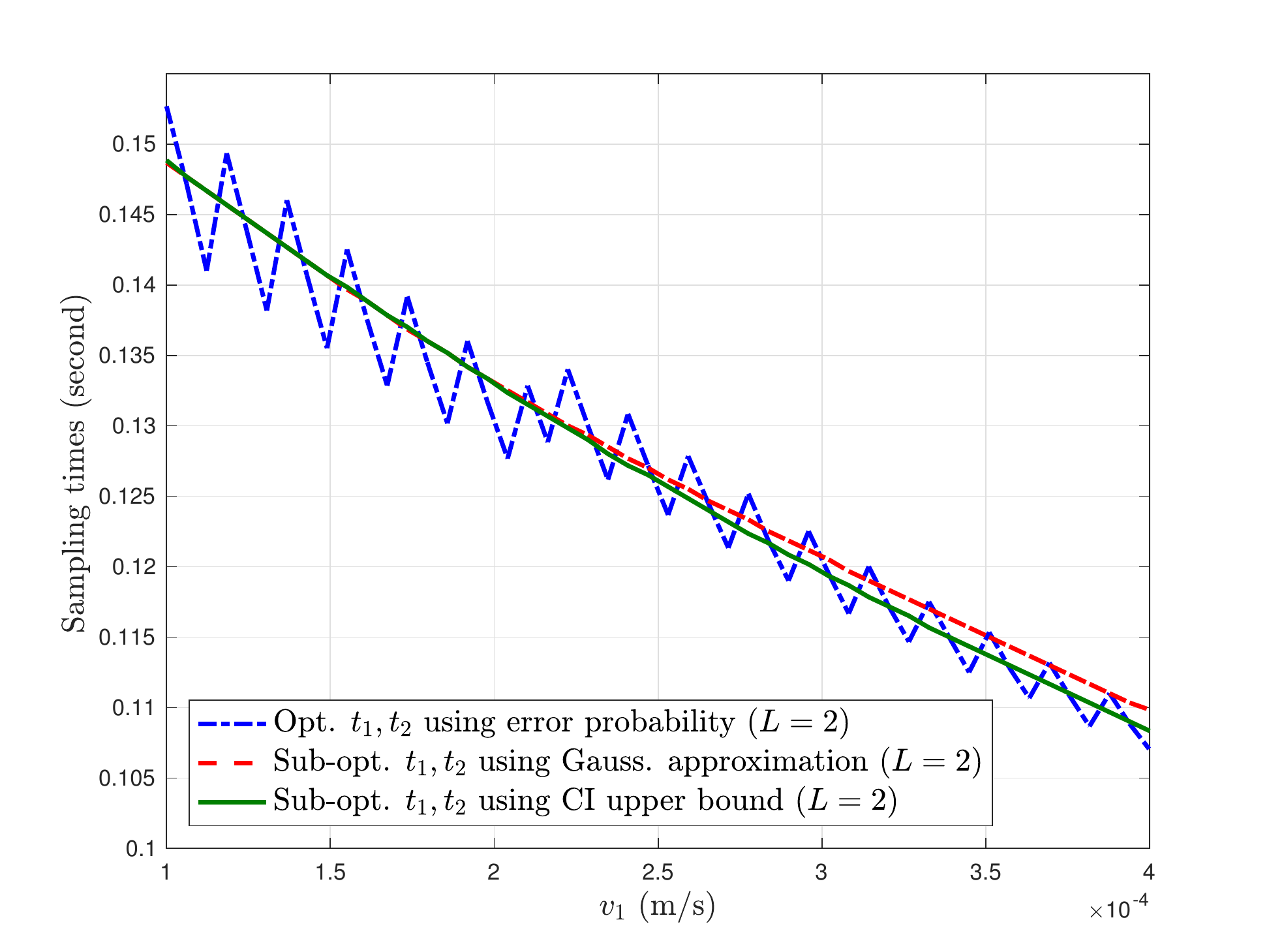}
\vspace{-1.5em}
\caption{Sampling times versus $v_1$ for $L=2$.}
\label{fig_tm_v0_lspb_ook_2}
\end{subfigure}
\caption{Error performance and sampling times for a flow velocity detector with $M=2$.}
\label{figtotalcapacity}
\vspace{-1.0em}
\end{figure*}
\subsubsection{Binary hypothesis ($M=2$)} 
The error probability and its Gaussian approximation for binary hypothesis with one-sample decoder, derived in \eqref{Pe_L_1} and \eqref{Pe_approx_Gauss_L_1}, respectively, and the CI upper bound, derived in \eqref{Peu_1sps_1}, are depicted in Fig. \ref{fig_pe_tm} versus the sampling time $t_1$ for $v_0=0, v_1=4\times 10^{-4}$ m/s.
The sampling times, which minimize the error probability, and its Gaussian approximation and CI upper bound are $t_{1}^*=0.1090$ s, $t_{1,\textrm{G}}=0.1097$ s, and $t_{1,\textrm{CI}}=0.1083$ s, respectively. It is seen that optimum and sub-optimum sampling times are nearly the same. 
We assume $v_0=0$ and depict the error probability, its Gaussian approximation, and CI upper bound for $L=1,2,3$ in Fig. \ref{fig_pe_v0} versus $v_1$ for their related optimum and sub-optimum sampling times. As expected, the error probability, the Gaussian approximation and the CI upper bound decrease as $v_1$ increases. Further, it is seen that the Gaussian approximation is nearly the same as the exact error probability. But, it makes distance as the error probability reduces. The CI upper bound and the error probability has a nearly constant gap in all values of $v_1$. 
The error probabilities and their Gaussian approximation and CI upper bound for $L=2,3$ have the same behavior as $L=1$ with the difference that they decrease as $L$ increases.\\
The sub-optimum sampling times, using CI upper bound and Gaussian approximation, given in Lemmas \ref{lemma_opt_tm_chernoff} and \ref{lemma_opt_tm_gauss}, along with the optimum sampling time using the exact error probability, are depicted versus $v_1$ in Fig. \ref{fig_tm}. It is seen that the sub-optimum sampling times are nearly the same, and decrease as $v_1$ increases and the optimum sampling time, which minimizes the error probability, fluctuates around the sub-optimum value (the fluctuation is small and because of the discrete nature of the Poisson distribution). 
In Fig. \ref{fig_tm_v0_lspb_ook_2}, the sampling times are depicted versus $v_1$ for a two-sample decoder. 
As mentioned in the previous section, the analytic results show that for $L$ sample decoder, the sampling times which minimize the CI upper bound are the same, which is verified by simulations i.e., $t_{1,\textrm{CI}}=...=t_{L,\textrm{CI}}$. It is seen using simulations that the $L$ sampling times which minimize the error probability and its Gaussian approximation are also equal, i.e., $t_1^*=...=t_L^*$ and $t_{1,\textrm{G}}=...=t_{L,\textrm{G}}$ for our simulation parameters. Further, it can be seen by comparing Fig. \ref{fig_tm} and Fig. \ref{fig_tm_v0_lspb_ook_2} that the optimum sampling times for $L=1,2$ are approximately equal.

\begin{figure*}
\centering
\begin{subfigure}[t]{0.23\textwidth}
\centering
\includegraphics[trim={1.5cm 0 0 0},scale=0.25]{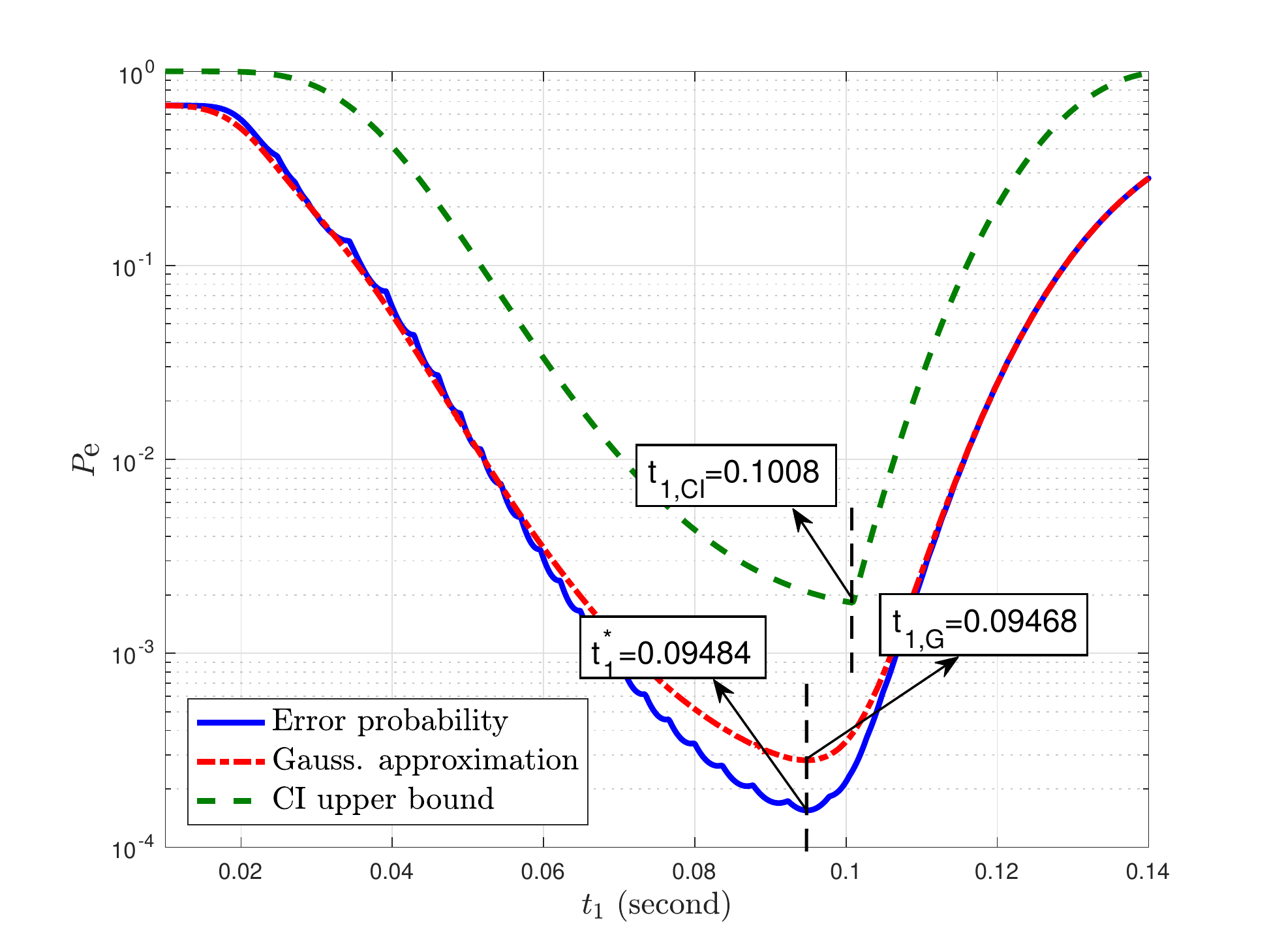}
\vspace{-1.5em}
\caption{Error performance versus $t_1$ for $L=1$.}
\label{fig_pe_tm_1sps_M3}
\end{subfigure}\quad 
\begin{subfigure}[t]{0.23\textwidth}
\centering
\includegraphics[trim={1.5cm 0 0 0},scale=0.25]{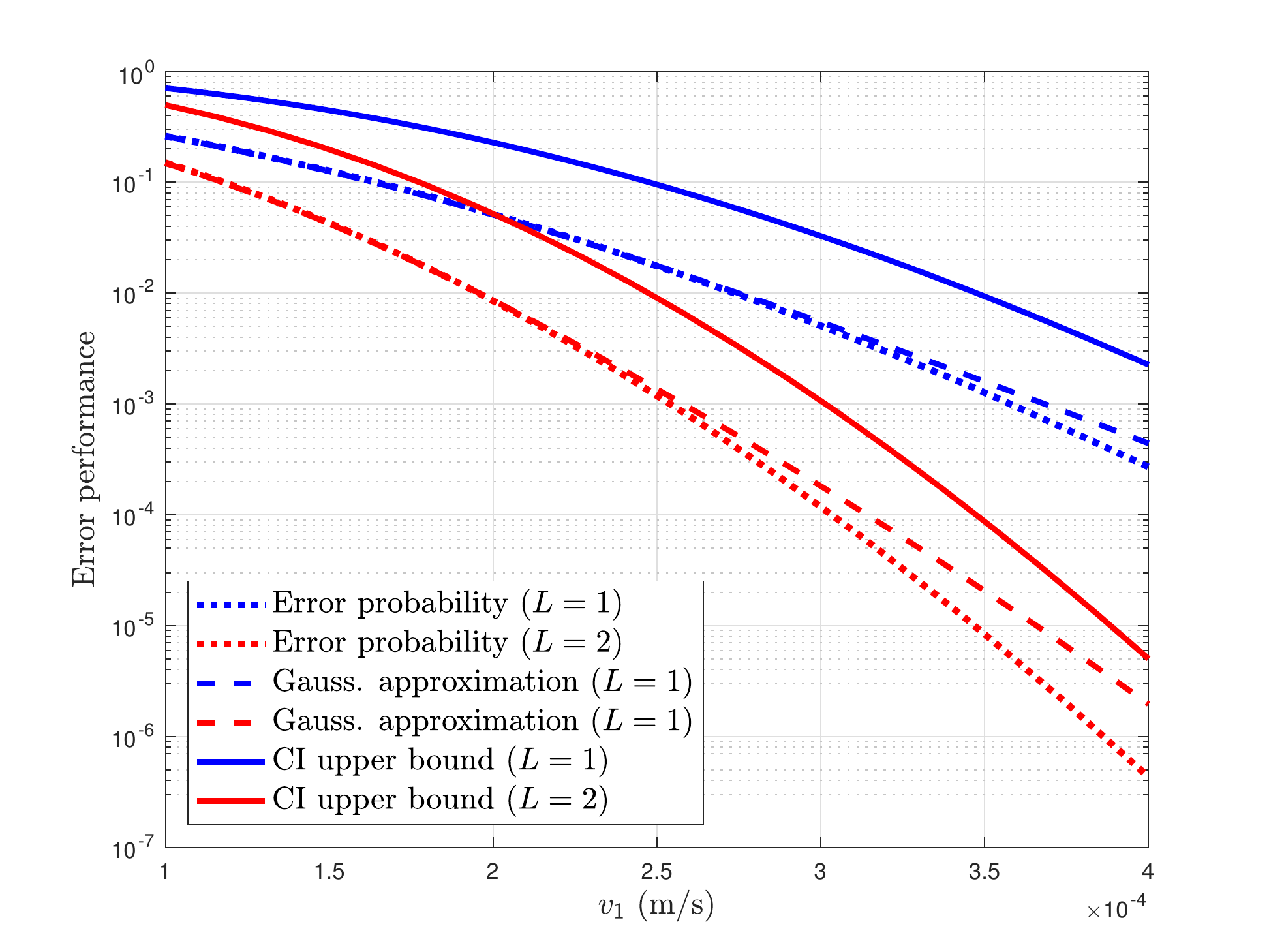}
\vspace{-1.5em}
\caption{Error performance versus $v_1$ for $L=1,2,3$ with related optimum sampling times.}
\label{fig_pe_v0_1sps_M3}
\end{subfigure}\quad
\begin{subfigure}[t]{0.23\textwidth}
\centering
\includegraphics[trim={1.5cm 0 0 0},scale=0.25]{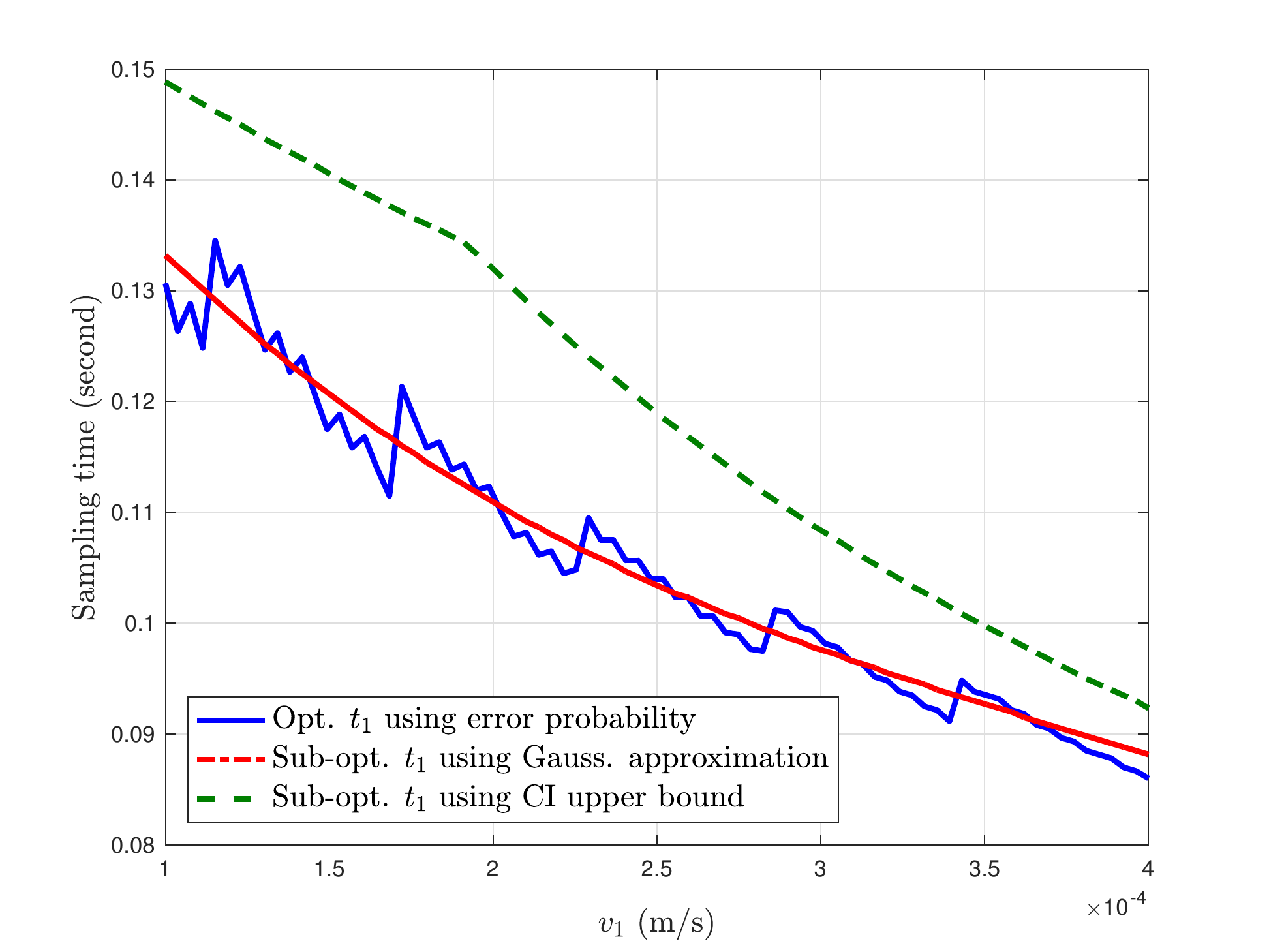}
\vspace{-1.5em}
\caption{Sampling times versus $v_1$ for $L=1$.}
\label{fig_tm_v0_1sps_M3}
\end{subfigure}\quad
\begin{subfigure}[t]{0.23\textwidth}
\centering
\includegraphics[trim={1.5cm 0 0 0},scale=0.25]{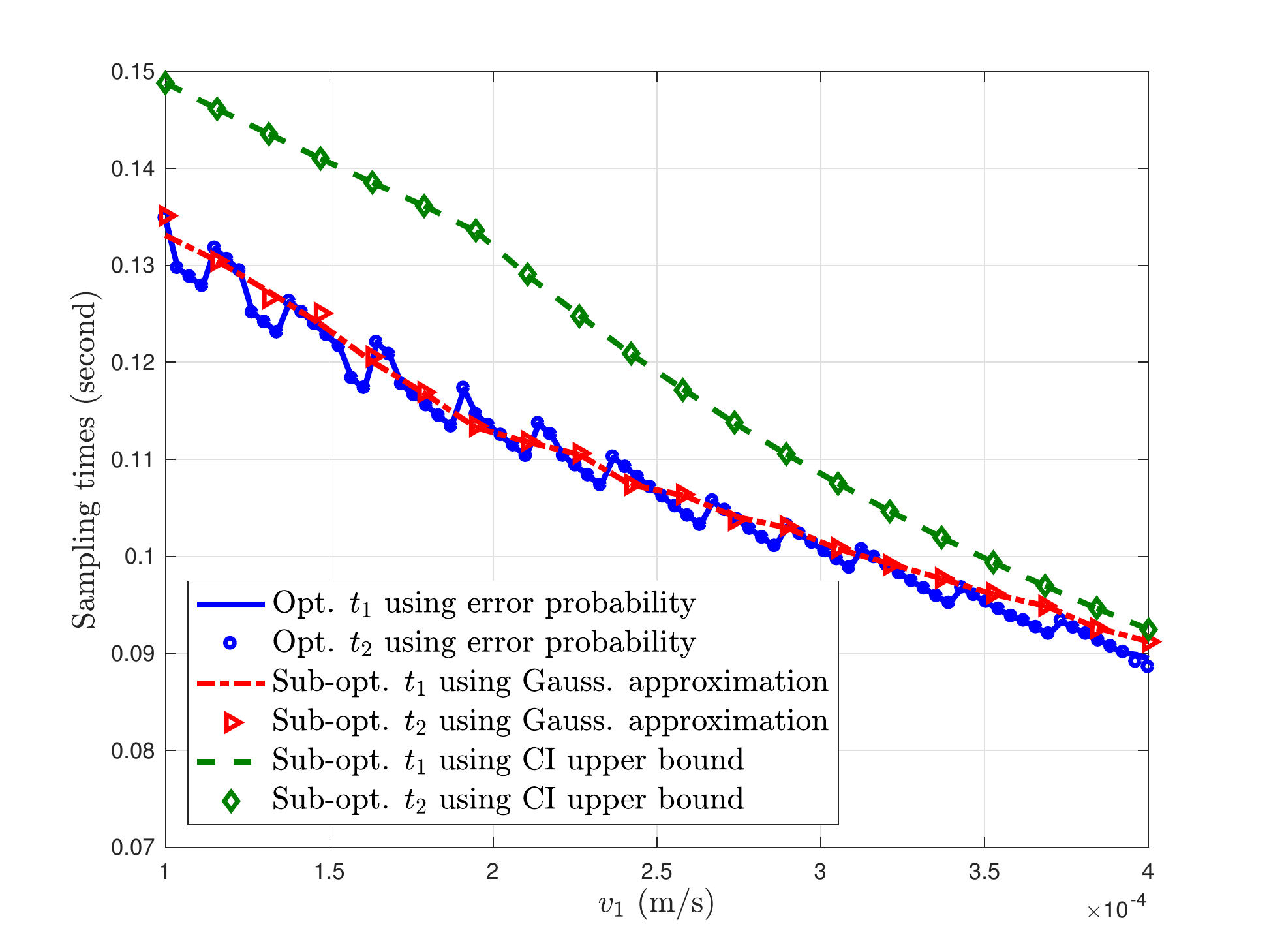}
\vspace{-1.5em}
\caption{Sampling times versus $v_1$ for $L=2$.}
\label{fig_tm_v0_2sps_M3}
\end{subfigure}
\caption{Error performance and sampling times for a flow velocity detector with $M=3$.}
\label{figtotalcapacity}
\vspace{-1.0em}
\end{figure*}
\subsubsection{Multiple hypotheses ($M>2$)}\label{simulation_ook}
Here, we assume $M=3$ hypotheses.
 The error probability, Gaussian approximation and CI upper bound versus the sampling time $t_1$ for $v_0=0$, $v_1=4\times10^{-4}$ m/s, and $v_2=8\times10^{-4}$ m/s are depicted Fig. \ref{fig_pe_tm_1sps_M3}. The sampling times that minimize the error probability, its Gaussian approximation and CI upper bound are obtained as $t_{1}^*=0.09484$ s, $t_{1,\textrm{G}}=0.09468$ s, and $t_{1,\textrm{CI}}=0.1008$ s, respectively. 
It is seen that the CI upper bound has a gap with the error probability in all values despite the binary case due to using union bound in multiple hypotheses case. The optimum value and the sub-optimum values of the sampling times are nearly equal. However, the sub-optimum value using CI upper bound has made a small gap from the optimum value compared to binary case, which may be due to the union bound. 
The error probability, the Gaussian approximation and the CI upper bound versus $v_1$ for the sampling times which minimize them are provided in Fig. \ref{fig_pe_v0_1sps_M3} for $L=1,2$. It is seen that as expected, the error probability, Gaussian approximation, and CI upper bound decrease as $L$ increases.\\
The sampling times which minimize the error probability, the Gaussian approximation, and the CI upper bound are depicted in Fig. \ref{fig_tm_v0_1sps_M3} versus $v_1$ (we assumed $v_0=0, v_2=2v_1$ and changed $v_1$). As seen in this figure, the optimum value of the sampling time fluctuates around the sub-optimum value using Gaussian approximation. But the sub-optimum value using CI upper bound has a distance from these values, which decreases as $v_1$ increases. For $L=2$, we depict the two sampling times $t_1$ and $t_2$ versus $v_1$ in Fig. \ref{fig_tm_v0_2sps_M3}. It is seen that similar to the binary case, the sampling times $t_1$ and $t_2$ which minimize the error probability, the Gaussian approximation, and the CI upper bound are equal, i.e., $t_1^*=t_2^*$, $t_{1,\textrm{G}}=t_{2,\textrm{G}}$, and $t_{1,\textrm{CI}}=t_{2,\textrm{CI}}$. Further, the optimum and sub-optimum sampling times are nearly the same as the values of the optimum and sub-optimum sampling times in one sample decoder.\\
For a large value of $L$ (e.g., $L=50$), we obtain the sub-optimum sampling times which minimize the CI upper bound using the optimization problem in \eqref{sub_opt_sampling_times}. We assume $v_0=0$, $v_1=10^{-4}$ m/s, and $v_2=2\times10^{-4}$ m/s. Using \eqref{sub_opt_sampling_times}, we obtain the $L=50$ sampling times as $t_1\approx...\approx t_{50}\approx0.1488$s. This is also confirmed by the results of Lemma \ref{lemma_opt_tm_inf_L} (which can be obtained from \eqref{optimization_problem_lemma5} equal to $t_{1, \textrm{Chernoff}}=t_{2, \textrm{Chernoff}}=t_{3, \textrm{Chernoff}}=0.1488$s). Note that Lemma \ref{lemma_opt_tm_inf_L} anticipates that the sub-optimum sampling times are {\em at most} ${M \choose 2}$ different times. This is somehow counter-intuitive since we obtain a single distinct sampling time, while we expect to obtain three different sampling times. For some other simulation parameters, the same result, i.e., a single sampling time, is observed. 
Another approach to find the three sampling times that Lemma \ref{lemma_opt_tm_inf_L} anticipates is to obtain the optimum times that discriminate between $\{H_0, H_1\}$, $\{H_1,H_2\}$, and $\{H_0, H_2\}$ by using one-sample decoder and minimizing the CI upper bound. For $\{H_0, H_1\}$, we get $t_1=0.1488$s, for $\{H_1,H_2\}$, we get $t_2=0.1231$s, and for $\{H_0, H_2\}$, we get $t_3=0.1330$s. Then, we use these sampling times in the optimization problem in Lemma \ref{lemma_opt_tm_inf_L} and obtain the three weights as $w_1^*=1, w_2^*=0, w_3^*=0$, which matches the results obtained from \eqref{sub_opt_sampling_times} and Lemma \ref{lemma_opt_tm_inf_L}. 

\begin{figure}
\centering
\begin{subfigure}[t]{0.23\textwidth}
\centering
\includegraphics[trim={1.5cm 0 0 0},scale=0.25]{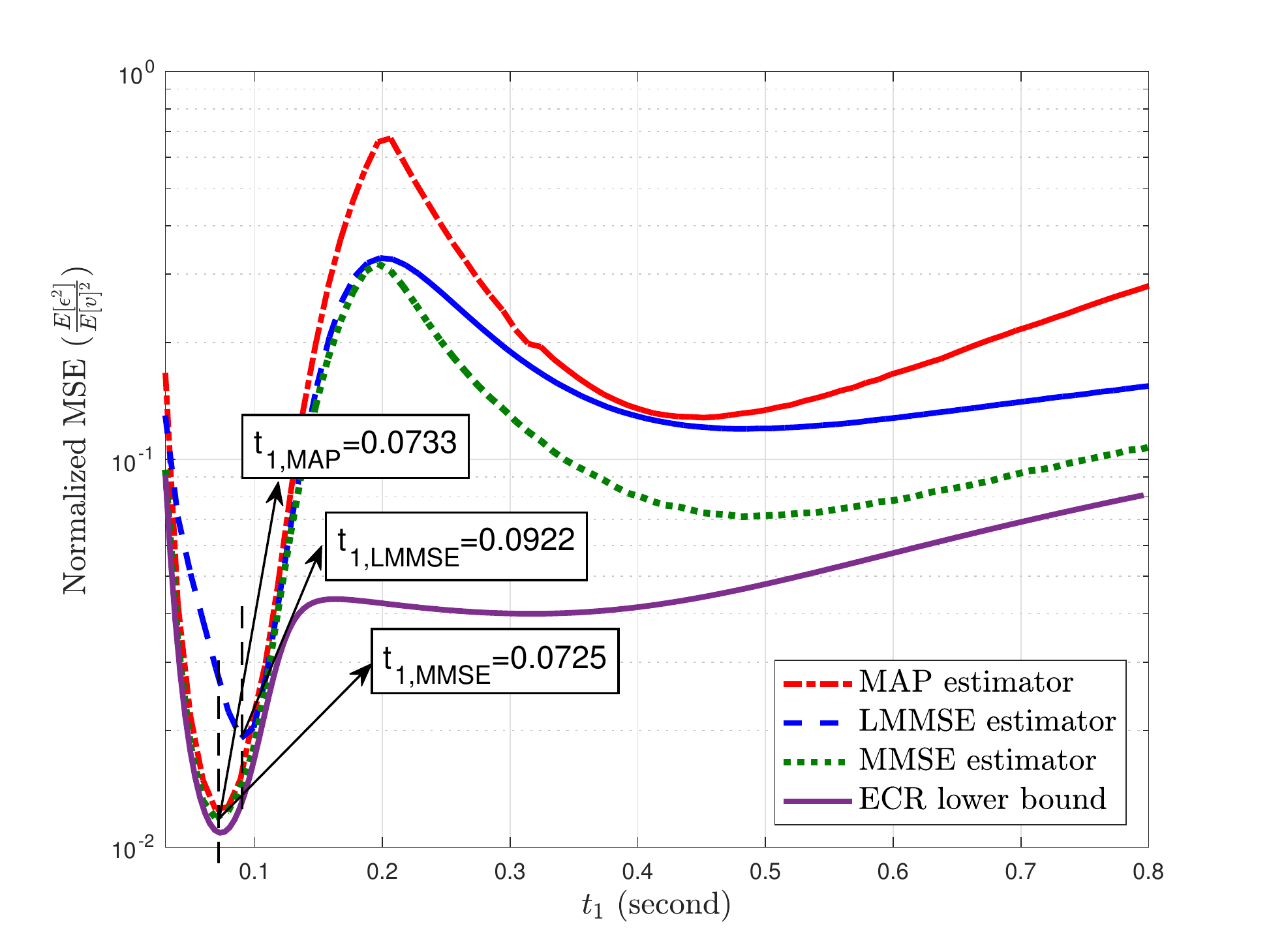}
\caption{Normalized MSE versus $t_1$.}
\label{fig_Eepsilon2_versus_t1}
\end{subfigure}\quad
\begin{subfigure}[t]{0.23\textwidth}
\centering
\includegraphics[trim={1.5cm 0 0 0},scale=0.25]{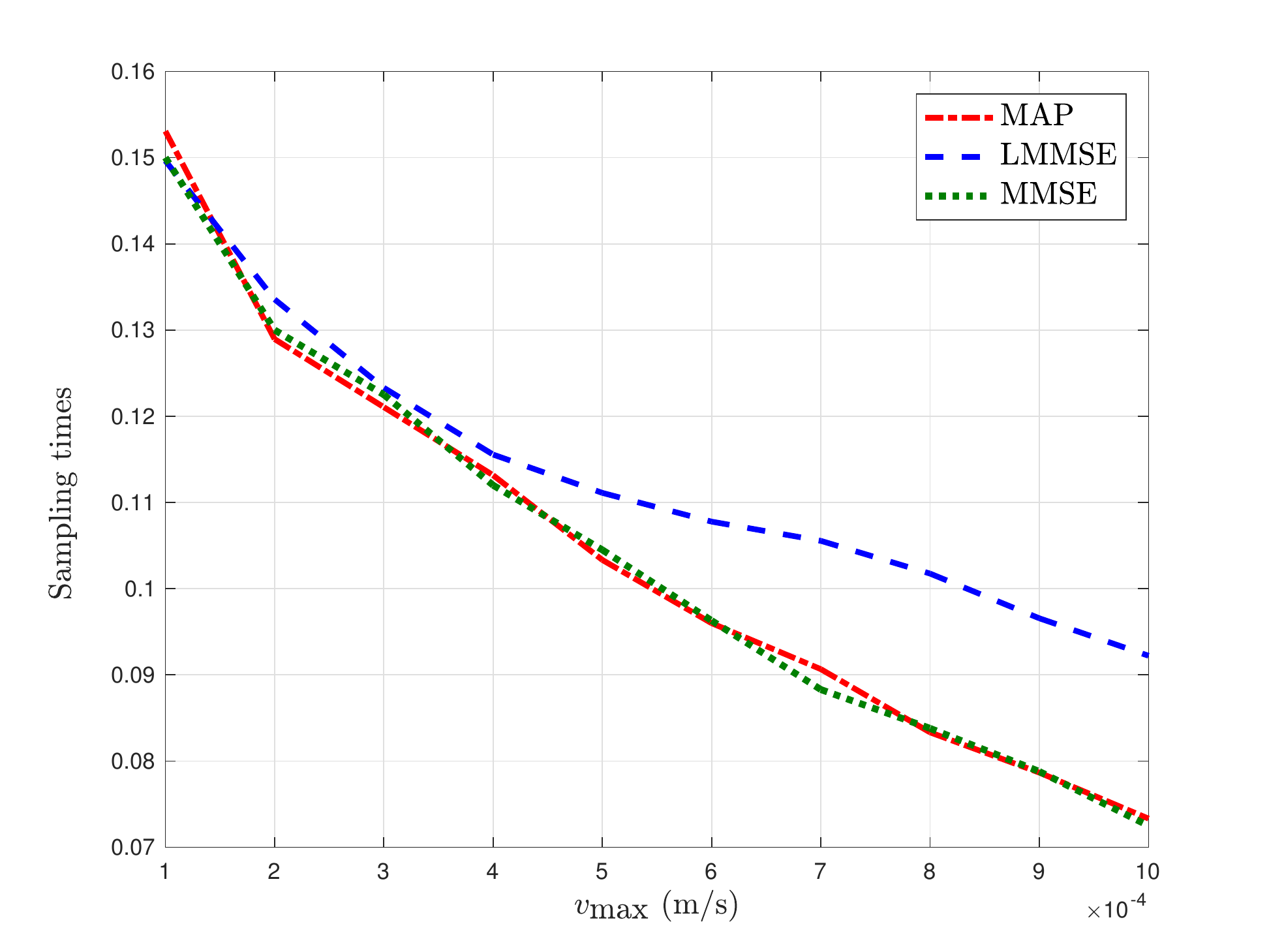}
\caption{Optimum sampling times versus $v_\textrm{max}$.}
\label{fig_tm_opt_versus_v1}
\end{subfigure}
\caption{Normalized MSE and optimum sampling times for a flow velocity estimator with $L=1$.}
\vspace{-1.5em}
\label{figtotalestimator}
\end{figure}

\subsection{Flow velocity estimator}\label{simulation_estimator}
Here, we assume that the magnitude of $\bm{v}$ has uniform distribution in the range $[v_\textrm{min}, v_\textrm{max}]$.
The MSE of estimation, $\E[\epsilon^2]$, normalized to $E[v]^2$, versus the sampling time is depicted in Fig. \ref{fig_Eepsilon2_versus_t1} for the MAP, MMSE, and LMMSE estimators along with the ECR lower bound given in Lemma \ref{lemma_ECR_est_L_performance}. 
We assume $v_\textrm{min}=0, v_\textrm{max}=10^{-3}$ m/s. As expected the MMSE estimator has the least MSE. However, the LMMSE does not always have better performance than MAP, which is because of the force of linearity to the estimated value in LMMSE. Further, it is seen that the performance of the MAP estimator is near the MMSE estimator and the ECR lower bound around the optimum sampling time.
Using this figure, the sampling times which minimize $\E[\epsilon^2]$ for the MAP, MMSE, and LMMSE estimators are obtained as $t_{1,\textrm{MAP}}=0.0733$ s, $t_{1,\textrm{MMSE}}=0.0725$ s, and $t_{1,\textrm{LMMSE}}=0.0922$ s, respectively. 
It is seen that the optimum sampling times of the MAP and MMSE are nearly the same.
However, the optimum sampling time of the LMMSE has a small gap from these values.
We assume $v_\textrm{min}=0$ and depict the optimum sampling times which minimize the MSE of the MAP,
MMSE, 
and LMMSE estimators
versus $v_\textrm{max}$ in Fig. \ref{fig_tm_opt_versus_v1}. It is seen that the optimum sampling times reduce as $v_\textrm{max}$ increases.

\section{Concluding Remarks and Future Works}\label{conclusion}
In this paper, we designed a molecular flow velocity meter which consists of a molecule releasing node and a molecular receiver to detect the medium flow velocity. 
We first assumed $M$ hypotheses for the medium flow velocity and a $L$-sample decoder at the receiver and obtained the optimum maximum-a-posteriori (MAP). We derived the error probability, its Gaussian approximation, and CI upper bound to analyze the performance of the detector. Further, we obtained the optimum and sub-optimum sampling times using the error probability, its Gaussian approximation, and CI upper bound. When $L \rightarrow \infty$, we obtained an interesting result using the CI upper bound which shows that for $M$ hypotheses, the sub-optimum sampling times yields to ${M \choose 2}$ sampling times $t_1, t_2,..., t_{M \choose 2}$ with ${M \choose 2}$ weights $w_1, w_2,...,w_{M \choose 2}$, i.e., $Lw_l$ sampling times are equal to $t_l$. For the simulation parameters, it is seen that these sampling times are the sampling times which minimize the CI upper bound for discriminating each two hypotheses. This results in a much simpler optimization problem to obtain the sub-optimum sampling times. Then, we assumed randomly chosen constant flow velocity in the medium and obtained the MAP and MMSE estimators for the $L$-sample receiver. 
We considered the mean square error (MSE) of the estimators and obtained the Bayesian Cramer-Rao (BCR) and expected Cramer-Rao (ECR) lower bounds on the MSE. We obtained the sampling times which minimize the MSE numerically. We showed that when $L \rightarrow \infty$, for the MAP estimator, two different sampling times are enough for estimation, i.e., $L{\tilde w_1}$ sampling times are $t_1$ and $L{\tilde w_2}$ sampling times are $t_2$. 
The molecular flow velocity meter can have applications in health care to monitor the function of the heart. It can also be used to design a new modulation scheme in MC, in which information is encoded in the medium flow velocity, i.e., similar to the classic communications that medium-based communication is introduced, we can introduce flow-based communication in MC. This makes the transmitter much simpler which is an important challenge in MC. 
\bibliographystyle{ieeetr}
\bibliography{reftest}

\appendices
\section{Proof of Lemma \ref{lemma_Pe_2}}\label{AppendixProoflemma2}
Let $R_i=\prod_{l=1}^{L} (\lambda_{i,l})^{y_j}\exp(-\lambda_{i,l})$. Using the optimum decision rule given in \eqref{decision_rule_csk}, the error probability can be obtained as
\begin{align}\label{eq_MAP_error}
\nonumber
&P_\textrm{e}
=\frac{1}{M}\sum_{i=0}^{M-1}\big[1-\Prob\{\cap_{j\in I, j \neq i} R_i>R_j |H_i\}\big]\\\nonumber
&=1-\frac{1}{M}\sum_{i=0}^{M-1}\sum^{\infty}_{\substack{y_1,...,y_L=0, \\R_i>R_j,~ j\in I,~ j \neq i}} \Prob(y_1,...,y_L|H_i)\\
&=1-\frac{1}{M}\sum_{i=0}^{M-1}\sum^{\infty}_{\substack{y_1,...,y_L=0, \\R_i>R_j,~ j\in I,~ j \neq i}} \prod_{l=1}^{L} \Prob(y_l|H_i).
\end{align}
The condition $R_i>R_j$ reduces to $w_{i,j,l} y_l>\beta_{i,j}$, where $w_{i,j,l}=\ln(\frac{\lambda_{i,l}}{\lambda_{j,l}})$ and $\beta_{i,j}=\sum_{l=1}^{L}(\lambda_{i,l}-\lambda_{j,l})$ (similar to \eqref{threshold2}), and hence, $P_\textrm{e}$ simplifies to
\begin{equation}
\begin{aligned}
\nonumber
P_\textrm{e}=1-\frac{1}{M}\sum_{i=0}^{M-1}\sum^{\infty}_{\substack{y_1,...,y_L=0, \\ \sum_{l=1}^{L} w_{i,j,l} y_l>\beta_{i,j},\\j\in I,~j \neq i}}\prod_{l=1}^{L} \Prob(Y_l=y_l|H_i).
\end{aligned}
\end{equation}
Now, by substituting the Poisson distribution for $\Prob(Y_l=y_l|H_i)$, we obtain \eqref{Pe_Poisson}.

\section{Proof of Corollary \ref{corollary_lemma_Pe_2_1}}
\label{AppendixProofcorollary3_1}
For binary hypothesis, \eqref{Pe_Poisson} is simplified as
\begin{align}
\nonumber
P_\textrm{e}&=1-\frac{1}{2}\big[\sum_{\substack{y_1,...,y_L=0,\\ \sum_{l=1}^{L} w_{0,1,l} y_l>\beta_{0,1}}}^{\infty}\prod_{l=1}^{L} \frac{(\lambda_{0,l})^{y_l} \exp(-\lambda_{0,l})}{y_l!}\\
&\quad+\sum_{\substack{y_1,...,y_L=0,\\ \sum_{l=1}^{L} w_{1,0,l} y_l>\beta_{1,0}}}^{\infty}\prod_{l=1}^{L} \frac{(\lambda_{1,l})^{y_l} \exp(-\lambda_{1,l})}{y_l!}\big)\big].
\end{align}
For this case, we have $w_{0,1,l}=-w_{1,0,l}=w_L$ and $\beta_{0,1}=-\beta_{1,0}=\beta$. Hence,
\begin{align}
\nonumber
P_\textrm{e}
&=\frac{1}{2}\big[1-\sum_{\substack{y_1,...,y_L=0,\\ \sum_{l=1}^{L} w_l y_l>\beta}}^{\infty}\prod_{l=1}^{L} \frac{(\lambda_{0,l})^{y_l} \exp(-\lambda_{0,l})}{y_l!}\\
&\quad+\sum_{\substack{y_1,...,y_L=0,\\ \sum_{l=1}^{L} w_l y_l>\beta}}^{\infty}\prod_{l=1}^{L} \frac{(\lambda_{1,l})^{y_l} \exp(-\lambda_{1,l})}{y_l!}\big)\big],
\end{align}
which reduces to \eqref{Pe_M_2}. For the Gaussian approximation, we have
\begin{align}\label{eq_pe_gauss}
\nonumber
P_\textrm{e}
&=\frac{1}{2}\sum_{i=0}^{M-1}\big[1-\Prob\{\sum_{l=1}^{L} w_l y_l>\beta |H_0\}\\
&\quad+\Prob\{\sum_{l=1}^{L} w_l y_l>\beta |H_1\}\big].
\end{align}
Since $Y_l$s are independent Gaussian variables, $Y=\sum_{l=1}^{L} w_l Y_l$ is a Gaussian variable with mean $\E[Y]=\sum_{l=1}^{L} w_l \lambda_{i,l}$ and variance $\V(Y)=\sum_{l=1}^{L} w_l^2 \lambda_{i,l}$, for $H_i$. Hence, \eqref{eq_pe_gauss} reduces to \eqref{Pe_approx_Gauss_M_2}.

\section{Proof of Lemma \ref{lemma_chernoff_UPe_Lsps_1}}
\label{AppendixProoflemma3}
The error probability of the MAP detector with $M$ hypotheses in \eqref{eq_MAP_error} can also be written as 
\begin{align}
P_{\textrm{e}}
&=\frac{1}{M}\sum_{i=0}^{M-1}\Prob\{\cup_{j\in I, j \neq i} R_i<R_j |H_i\}.
\end{align}
Now, we upper bound the error probability as follows:
\begin{align}
&P_{\textrm{e}}\overset{(a)}\leq \frac{1}{M}\sum_{i=0}^{M-1}\sum_{\substack{j=0,\\ j \neq i}}^{M-1}\Prob\{R_i<R_j |H_i\}\\\nonumber 
&~=\frac{1}{M}\sum_{i=0}^{M-1}\sum_{\substack{j=0,\\j \neq i}}^{M-1}\sum_{y_1,...,y_L=0}^{\infty}1\big\{\Prob(y_1,...,y_L|H_{i})\\\nonumber
&~\quad <\Prob(y_1,...,y_L|H_{j})\big\}\Prob(y_1,...,y_L|H_{i})\\\nonumber
&~=\frac{1}{M}\sum_{\substack{i_1,i_2=0, \\ i_1 \neq i_2}}^{M-1}\sum_{y_1,...,y_L=0}^{\infty}\min\big\{\Prob(y_1,...,y_L|H_{i_1}),\\\nonumber
&~\quad \Prob(y_1,...,y_L|H_{i_2})\big\}\\\nonumber
&~\overset{(b)}\leq \frac{1}{M}\sum_{\substack{i_1,i_2=0, \\ i_1 \neq i_2}}^{M-1}\sum_{y_1,...,y_L=0}^{\infty}\Prob(y_1,...,y_L|H_{i_1})^{s_{i_1,i_2}}\\\nonumber
&~\quad \times\Prob(y_1,...,y_L|H_{i_2})^{1-s_{i_1,i_2}}\\\nonumber 
&~\overset{(c)}= \frac{1}{M}\sum_{\substack{i_1,i_2=0, \\ i_1 \neq i_2}}^{M-1}\prod_{l=1}^{L}\sum_{y_l=0}^{\infty}\Prob(y_l|H_{i_1})^{s_{i_1,i_2}}\Prob(y_l|H_{i_2})^{1-s_{i_1,i_2}}\\\nonumber 
&~\overset{(d)}= \frac{1}{M}\sum_{\substack{i_1,i_2=0, \\ i_1 \neq i_2}}^{M-1}\exp(-D_{i_1,i_2}(s_{i_1,i_2}))\\\nonumber
&~\overset{(e)}\leq \frac{1}{M}{M \choose 2}\max_{\substack{i_1,i_2 \in I, \\ i_1 \neq i_2}}\exp(-D_{i_1,i_2}(s_{i_1,i_2})), 
\end{align}
where $D_{i_1,i_2}(s_{i_1,i_2})=\sum_{l=1}^{L} [\lambda_{i_1,l}s_{i_1,i_2}+\lambda_{i_2,l}(1-s_{i_1,i_2})-\lambda_{i_1,l}^{s_{i_1,i_2}} \lambda_{i_2,l}^{1-s_{i_1,i_2}}]$, (a) is due to the union bound, (b) is due to eq. \eqref{chernof_ineq}, (c) is due to assuming independent observations at the receiver, (d) is due to Poisson distribution for observations, i.e., $\Prob(y_l|H_i)=\frac{(\lambda_{i,l})^{y_l}\exp(-\lambda_{i,l})}{y_l!}$, $i\in I$, $l \in \{1,...,L\}$, and (e) is due to substituting the maximum term for each term of the summation. Since the bound holds for all values of $s_{i_1,i_2} \in (0,1)$, it also holds for the optimum values of $s_{i_1,i_2}$, which is obtained by minimizing $\exp\big(-D_{i_1,i_2}(s_{i_1,i_2})\big)$, i.e., maximizing $D_{i_1,i_2}(s_{i_1,i_2})$ with respect to $s_{i_1,i_2}$. 
Hence, we obtain the upper bound as \eqref{Peu_Lsps_1}. 
The optimum $s_{i_1,i_2}$ is obtained by minimizing $\exp\big(-D_{i_1,i_2}(s_{i_1,i_2})\big)$, i.e., maximizing $D_{i_1,i_2}(s_{i_1,i_2})$, with respect to $s_{i_1,i_2}$. Hence $s_{i_1,i_2}^*$ is the solution of $\frac{d}{ds_{i_1,i_2}}D_{i_1,i_2}(s_{i_1,i_2})=0$, which can be simplified as \eqref{opt_s_Lsps_1}. 

\section{Proof of Corollary \ref{corollary_chernoff_UPe_Lsps_1_1}}
\label{AppendixProofcorollary4_1}
From holder's inequality, for any positive vectors $\bm{x}=(x_1,x_2,...,x_n)$ and $\bm{y}=(y_1,y_2,...y_n)$ and for any $p,q$, satisfying $p>1$ and $\frac{1}{p}+\frac{1}{q}=1$, we have $(\sum_{i=1}^{n} x_i^p)^{\frac{1}{p}}(\sum_{i=1}^{n} y_i^q)^{\frac{1}{q}}\geq \sum_{i=1}^{n} x_i y_i$. Using this inequality for $\bm{x}=(\lambda_{i_1,1}^{s_{i_1,i_2}},...,\lambda_{i_1,L}^{s_{i_1,i_2}}), \bm{y}=(\lambda_{i_2,1}^{1-s_{i_1,i_2}},...,\lambda_{i_2,L}^{1-s_{i_1,i_2}})$, $p=\frac{1}{s_{i_1,i_2}}$, and $q=\frac{1}{1-s_{i_1,i_2}}$, we have:
\begin{equation*}
\begin{aligned}
\sum_{l=1}^{L}\lambda_{i_1,l}^{s_{i_1,i_2}} \lambda_{i_2,l}^{1-s_{i_1,i_2}} \leq (\sum_{l=1}^{L}\lambda_{i_1,l})^{s_{i_1,i_2}} (\sum_{l=1}^{L}\lambda_{i_2,l})^{1-s_{i_1,i_2}}.
\end{aligned}
\end{equation*}
Hence, we bound $D_{i_1,i_2}(s_{i_1,i_2})$ in \eqref{Peu_Lsps_1} as follows:
\begin{align}
\nonumber
D_{i_1,i_2}(s_{i_1,i_2}) &\geq (\sum_{l=1}^{L} \lambda_{i_1,l})s_{i_1,i_2}+(\sum_{l=1}^{L}\lambda_{i_2,l})(1-s_{i_1,i_2})\\
&~-(\sum_{l=1}^{L}\lambda_{i_1,l})^{s_{i_1,i_2}} (\sum_{l=1}^{L}\lambda_{i_2,l})^{1-s_{i_1,i_2}}.
\end{align}
Let $K_{i_1,i_2}(s_{i_1,i_2})=(\sum_{l=1}^{L} \lambda_{i_1,l})s_{i_1,i_2}+(\sum_{l=1}^{L}\lambda_{i_2,l})(1-s_{i_1,i_2})-(\sum_{l=1}^{L}\lambda_{i_1,l})^{s_{i_1,i_2}} (\sum_{l=1}^{L}\lambda_{i_2,l})^{1-s_{i_1,i_2}}$. Then, \eqref{Peu_Lsps_1} reduces to \eqref{Peu_Lsps_2}. Now using this bound, the optimum value of $s_{i_1,i_2}$ is obtained by maximizing $K_{i_1,i_2}(s_{i_1,i_2})$ as the solution of the following equation:
\begin{align}
&\frac{d }{d s_{i_1,i_2}} K_{i_1,i_2}(s_{i_1,i_2})=\sum_{l=1}^{L} \lambda_{i_1,l}-\sum_{l=1}^{L}\lambda_{i_2,l}\\\nonumber
&\quad-\ln (\sum_{l=1}^{L} \lambda_{i_1,l})(\sum_{l=1}^{L} \lambda_{i_1,l})^{s_{i_1,i_2}}(\sum_{l=1}^{L}\lambda_{i_2,l})^{1-s_{i_1,i_2}}\\\nonumber &\quad +
\ln (\sum_{l=1}^{L} \lambda_{i_2,l})(\sum_{l=1}^{L}\lambda_{i_1,l})^{s_{i_1,i_2}}(\sum_{l=1}^{L}\lambda_{i_2,l})^{1-s_{i_1,i_2}}=0,
\end{align}
which reduces to
\small
$$s_{i_1,i_2}\ln(\frac{\sum_{l=1}^{L} \lambda_{i_1,l}}{\sum_{l=1}^{L}\lambda_{i_2,l}})+\ln \ln(\frac{\sum_{l=1}^{L} \lambda_{i_1,l}}{\sum_{l=1}^{L}\lambda_{i_2,l}})=\ln(\frac{\sum_{l=1}^{L} \lambda_{i_1,l}}{\sum_{l=1}^{L}.\lambda_{i_2,l}}-1).$$
\normalsize
Hence, $s_{i_1,i_2}^*$ is obtained as \eqref{opt_s_Lsps_2}.

\section{Proof of Lemma \ref{lemma_opt_tm_gauss}}
\label{AppendixProoflemma6}
From \eqref{sub_opt_sampling_times_G}, we should solve the optimization problem 
$\max_{t_{1},...,t_{L}} P_\textrm{e,G},$
where $P_\textrm{e,G}$ is given in \eqref{Pe_approx_Gauss_M_2}. Hence, the sub-optimum values of $t_{1},...,t_{L}$ are the solutions of $\nabla P_\textrm{e,G}=[\frac{\partial P_\textrm{e,G}}{\partial t_{1}}, ...,\frac{\partial P_\textrm{e,G}}{\partial t_{L}}]$ $=0$. Let $\mu_i=\sum_{l=1}^{L}w_l\lambda_{i,l}$ and $\sigma_i=\sqrt{\sum_{l=1}^{L} w_l^2\lambda_{i,l}}$ in \eqref{Pe_approx_Gauss_M_2}. Hence, from $\frac{\partial P_\textrm{e,G}}{\partial t_{l}}=0$, we obtain:
\begin{align}\label{setofeqs_tms_opt}
&e^{\frac{-(\beta-\mu_0)^2}{2\sigma_0^2}} \big(\frac{(\frac{d}{d t_l}\beta-\frac{d}{d t_l}\mu_0)\sigma_0-(\frac{d}{d t_l}\sigma_0)(\beta-\mu_0)}{\sigma_0^{2}}\big)\\\nonumber
&-e^{\frac{-(\beta-\mu_1)^2}{2\sigma_1^2}} \big(\frac{(\frac{d}{d t_l}\beta-\frac{d}{d t_l}\mu_1)\sigma_1-(\frac{d}{d t_l}\sigma_1)(\beta-\mu_1)}{\sigma_1^{2}}\big)=0.
\end{align}
Let $g_{i,l}=\frac{d}{dt_l}\lambda_{i,l}$. Hence, from definition of $\beta$ and $w_l$ in Corollary \ref{corollary_lemma_threshold_L_1}, we have $\frac{d}{d t_l}\beta=g_{0,l}-g_{1,l}$, and 
\begin{align}
&\frac{d}{d t_l}\mu_i=(\frac{g_{0,l}}{\lambda_{0,l}}-\frac{g_{1,l}}{\lambda_{1,l}})\lambda_{i,l}+w_l g_{i,l}, \quad i=0,1\\\nonumber
&\frac{d}{d t_l}\sigma_i=\frac{w_l}{2\sigma_i} \big[2(\frac{g_{0,l}}{\lambda_{0,l}}-\frac{g_{1,l}}{\lambda_{1,l}})\lambda_{i,l}+w_l g_{i,l}\big], \quad i=0,1.
\end{align}
Hence, we obtain the set of equations in \eqref{setofeqs_tms_opt_2}. For the location invariant flow velocity and transparent receiver, $\lambda_{i,l}=V_\textrm{R}\zeta h_0(\bm{r}_0-\int_{t_{0}}^{t_{1}} \bm{v}_i(\tau) d\tau,t)$. Using \eqref{eq_h0}, we have
\begin{align}
\nonumber
g_{i,l}&=V_\textrm{R}\zeta \frac{d}{d t_{l}} h_0(\bm{r}_0-\int_{t_\textrm{r}}^{t_{l}} \bm{v}_i(\tau) d\tau,t_{l})\\\nonumber
&=V_\textrm{R}\zeta  \frac{d}{d t_{l}} \big[\frac{1}{(4\pi D (t_l-t_\textrm{r}))^{\frac{3}{2}}}e^{-\frac{||\bm{r}_0-\int_{t_\textrm{r}}^{t_{l}}\bm{v}_i(\tau)d\tau||^2}{4D(t_l-t_\textrm{r})}}\bigg]\\\nonumber
&=\lambda_{i,l}\bigg[\frac{-3}{2(t_{l}-t_\textrm{r})}+\frac{\langle \bm{v}_i(t_l),\bm{r}_0-\int_{t_\textrm{r}}^{t_{l}}\bm{v}_i(\tau)d\tau \rangle }{2D(t_l-t_\textrm{r})}\\
&\quad+\frac{||\bm{r}_0-\int_{t_\textrm{r}}^{t_{l}}\bm{v}_i(\tau)d\tau||^2}{4D(t_l-t_\textrm{r})^2} \bigg].
\end{align}

\section{Proof of Lemma \ref{lemma_opt_tm_chernoff}}
\label{AppendixProoflemma7}
To obtain the sub-optimum sampling times using \eqref{Peu_Lsps_2}, we should solve 
\begin{align}
\max_{t_{1},...,t_{L}} \max_{s}  \sum_{l=1}^{L}[\lambda_{0,l}s+\lambda_{1,l}(1-s)-\lambda_{0,l}^{s} \lambda_{1,l}^{1-s}].
\end{align}
Let $f(t_l,s)=\lambda_{0,l}s+\lambda_{1,l}(1-s)-\lambda_{0,l}^{s} \lambda_{1,l}^{1-s}$. Hence, the optimum values of $t_{1},...,t_{L}$, and $s$ are the solutions of $\nabla \sum_{l=1}^{L}f(t_l,s)=[\frac{\partial f(t_{1},s)}{\partial t_{1}}, ..., \frac{\partial f(t_{L},s)}{\partial t_{L}},\sum_{l=1}^{L}\frac{\partial f(t_{l},s)}{\partial s}]=0$. From $\frac{\partial f(t_{l},s)}{\partial t_{l}}=0, l=1,...,L$, we conclude that $t_{1}=...=t_{L}$. Hence, from $\sum_{l=1}^{L}\frac{\partial f(t_{l},s)}{\partial s}=0$, we obtain $\frac{\partial f(t_{1},s)}{\partial s}=0$. Thus, $s$ and $t_{1}$ are the solutions of $[\frac{\partial f(t_{1},s)}{\partial t_{1}},\frac{\partial f(t_{1},s)}{\partial s}]=0$. This means that the sub-optimum values of $t_{1},...,t_{L}$ are equal to the values when $L=1$.\\
Now, we obtain the set of equations for the sub-optimum sampling time when $L=1$, i.e., $t_1$. 
$\frac{\partial f(t_{1},s)}{\partial s}=0$, results in $s=\frac{\ln(\frac{\lambda_{0,1}}{\lambda_{1,1}}-1)-\ln\ln(\frac{\lambda_{0,1}}{\lambda_{1,1}})}{\ln(\frac{\lambda_{0,1}}{\lambda_{1,1}})}$, and $\frac{\partial f(t_{1},s)}{\partial t_{1}}=0$, yields to
\small
\begin{align}
\nonumber
g_{0,1}s+g_{1,1}(1-s)&-s g_{0,1} (\frac{\lambda_{1,1}}{\lambda_{0,1}})^{1-s}-(1-s) g_{1,l} (\frac{\lambda_{0,1}}{\lambda_{1,1}})^{s}=0,
\end{align}
\normalsize
where $g_{i,l}$ is defined in Lemma \ref{lemma_opt_tm_gauss}.

\section{Proof of Corrolary \ref{corollary_MAP_est_L}}
\label{AppendixProofcorrollary8_1}
For $\bm{v}=v\bm{d}$, and uniform distribution for $v$, the MAP estimation of $v$ is 
\begin{align}\label{estimation_rule_csk2}
\hat{v}=\arg\max_{v \in S_v} y_1 \ln (\lambda_{1}(v\bm{d}))-\lambda_{1}(v\bm{d}).
\end{align}
In the following, let $R_\textrm{est,u}(v)=y_1 \ln (\lambda_{1}(v\bm{d}))-\lambda_{1}(v\bm{d})$. Hence, we should find the solutions of $R_\textrm{est,u}^{'}(v)=\frac{d}{dv}R_\textrm{est,u}(v)=0$ which maximize $R_\textrm{est,u}(v)$ and fall in $S_v$. For the transparent receiver, we have $R_\textrm{est,u}^{'}(v)=\frac{1}{2D}(r_0-v.(t_1-t_\textrm{r})).(y_1-\lambda_{1}(v\bm{d}))=0$.
If there is no maximizer in this range, we should consider $v_\textrm{min},v_\textrm{max}$. Hence, the candidates of the maximizer are the values of $v$ which satisfiy the equations $r_0-v.(t_1-t_\textrm{r})=0$ and $\lambda_{1}(v\bm{d})=y_1$. From $r_0-v.(t_1-t_\textrm{r})=0$, we obtain $v_1=\frac{r_0}{t_1-t_\textrm{r}}$. Note that $v_1$ maximizes $\lambda_{1}(v\bm{d})$, i.e., $\lambda_{1}(v\bm{d})$ is a positive function with maximum $\lambda_{1}(v_1\bm{d})=\frac{\zeta V_R}{(4\pi D (t_1-t_\textrm{r}))^{\frac{3}{2}}}$. For the second equation, we have three cases:

Case 1) $y_1=\lambda_{1}(v_1\bm{d})$: In this case, the only solution of the equation $\lambda_l(v_1\bm{d})=y_1$ is equal to $v_1$.

Case 2) $y_1>\lambda_{1}(v_1\bm{d})$: In this case, the second equation $\lambda_l(v\bm{d})=y_1$ does not have any solutions for $\bm{v}$ since $y_l$ is greater than the maximum value of $\lambda_l(v\bm{d})$.

Case 3) $y_1<\lambda_{1}(v_1\bm{d})$: In this case, the equation $\lambda_l(v\bm{d})=y_1$ has two solutions 
as $v_2=\frac{r_0+\sqrt{\Delta}}{t_l-t_\textrm{r}}$ and $v_3=\frac{r_0-\sqrt{\Delta}}{t_l-t_\textrm{r}}$, where $\Delta=-4D (t_1-t_\textrm{r})(\ln{y_l}-\ln{(\lambda_{1}(v_1\bm{d}))})$. In Case 1, we have 
$R_\textrm{est,u}^{'}(v)=\frac{1}{2D}(r_0-v.(t_1-t_\textrm{r})).(\lambda_{1}(v_1\bm{d})-\lambda_{1}(v\bm{d}))$. Since $v_1$ is the maximizer of $\lambda_{1}(v_1\bm{d})$, $\lambda_{1}(v_1\bm{d})-\lambda_{1}(v\bm{d})$ is positive for all $v \neq v_1$. Hence, for $v<v_1$, $R_\textrm{est,u}^{'}(v)>0$ and for $v>v_1$, $R_\textrm{est,u}^{'}(v)<0$, and thus, $v_1$ is the maximizer. Now, using the second derivative of $R_\textrm{est,u}(v)$, we show that in Case 2, $v_1$ is the only maximizer of $R_\textrm{est,u}(v)$ and in Case 3, $v_1$ is the minimizer and $v_2$ and $v_3$ are the maximizers of $R_\textrm{est,u}(v)$, and if both values fall in $S_v$, the estimator gives one of the values $v_2$ and $v_3$ randomly as the estimated value of $v$.\\
The second derivative of $R_\textrm{est,u}(v)$ can be obtained as
\begin{align}\label{eq_second_der}
\nonumber
R_\textrm{est,u}^{''}(v)&=\frac{d^2}{dv^2} R_\textrm{est,u}(v)=-\frac{1}{2D}(t_1-t_\textrm{r})(y_1-\lambda_{1}(v\bm{d}))\\
&\quad-\frac{1}{4D^2}(r_0-v.(t_1-t_\textrm{r}))^2\lambda_{1}(v\bm{d}).
\end{align}
 For $v=v_1$, we have $r_0-v.(t_1-t_\textrm{r})=0$ and hence,
\begin{align}
R_\textrm{est,u}^{''}(v)=-(t_1-t_\textrm{r})(y_1-\lambda_{1}(v\bm{d})).
\end{align}
Now in Case 2, we have $R_\textrm{est,u}^{''}(v_1)<0$, and hence, $\bm{v}_1$ is the maximizer. 
In Case 3, we have $R_\textrm{est,u}^{''}(v_1)>0$, and hence, $\bm{v}_1$ is the minimizer.\\
For $v_2$ and $v_3$ in Case 3, we have $y_1=\lambda_{1}(v\bm{d})$, and hence,
\begin{align}
R_\textrm{est,u}^{''}(v)=-\frac{1}{4D^2}(r_0-v.(t_1-t_\textrm{r}))^2\lambda_{1}(v\bm{d}).
\end{align}
Since $(\bm{r}_0-v.(t_l-t_\textrm{r}))^2>0$ and $\lambda_l(\bm{v})$ is a positive function, $R_\textrm{est,u}^{''}(v)|_{v=v_1,v_2}<0$ and hence, $\bm{v}_2$ and $\bm{v}_3$ are the maximizers. 

\section{Proof of Lemma \ref{lemma_MMSE_est_L}}
\label{AppendixProoflemma9}
For the MMSE estimator, when the mean and variance of $v_i, i \in \{x,y,z\}$ is finite, we have $\hat{v}_i=\E[v_i|y_1,...,y_L]$ \cite{van2004detection}.
For $\hat{v}_x$, we have
\small
\begin{align}
&\Prob(v_x|y_1,...y_L)=\frac{\Prob(y_1,...,y_L|v_x)p_x(v_x)}{\Prob(y_1,...,y_L)}\\\nonumber
&=\frac{p_x(v_x)\int_{v_y, v_z}p_y(v_y)p_z(v_z)\Prob(y_1,...,y_L|v_x,v_y,v_z)d v_z d v_y}{\int p_x(v_x)p_y(v_y)p_z(v_z)\Prob(y_1,...,y_L|v_x,v_y,v_z) dv_z dv_y dv_x}.
\end{align}
\normalsize
For the independent observations, we have
\begin{align}
&\Prob(v_x|y_1,...y_L)=\\\nonumber
&\quad\frac{p_x(v_x)\int_{v_y,v_z}p_y(v_y)p_z(v_z)\Pi_{l=1}^{L}\Prob(y_l|v_x,v_y,v_z)d v_z d v_y }{\int p_x(v_x)p_y(v_y)p_z(v_z)\Pi_{l=1}^{L}\Prob(y_l|v_x,v_y,v_z) dv_z dv_y dv_x}.
\end{align}
Therefore,
\small
\begin{align}
&\hat{v}_x=E(v_x|y_1,...y_L)=\\\nonumber
&\frac{\int_{v_x} v_x p_x(v_x)\int_{v_y,v_z}p_y(v_y)p_z(v_z)\Pi_{l=1}^{L}\Prob(y_l|\bm{v})d v_z d v_y d v_x }{\int p_x(v_x)p_y(v_y)p_z(v_z)\Pi_{l=1}^{L}\Prob(y_l|\bm{v}) dv_z dv_y dv_x}.
\end{align}
\normalsize
Let $p_{x,y,z}(\bm{v})=p_x(v_x)p_y(v_y)p_z(v_z)$. Since the conditional probability distribution of $Y_l, l=1,...,L$ given $\bm{v}$ is Poiss($\lambda_l(\bm{v})$), we obtain the equation \eqref{estimation_rule_MMSE}. \\
For the LMMSE, the estimator of $\bm{v}$ can be obtained as
$\hat{\bm{v}}=\bm{y}A+\bm{b},$
where $A=C_{YY}^{-1}C_{Y\bm{v}}, \bm{y}=(y_1,y_2,...,y_L), \bm{b}=\bm{\mu}-\bm{\lambda} A$, in which
\begin{align}
&\bm{\mu}=\E[\bm{v}]=(\E[v_x],\E[v_y],\E[v_z]),\\\nonumber
& \bm{\lambda}=(\E[Y_1],\E[Y_2],...,\E[Y_L]),
\end{align}
\small
\begin{align}
&C_{YY}=
\begin{bmatrix}
\V(Y_1) & \C(Y_1,Y_2) & .... & \C(Y_1,Y_L)\\
\C(Y_2,Y_1) & \V(Y_2) & .... & \C(Y_2,Y_L)\\
\vdots & \vdots &  & \vdots\\\nonumber
\C(Y_L,Y_1) & \C(Y_L,Y_2) & .... & \V(Y_L)
\end{bmatrix},\\\nonumber
& C_{Y\bm{v}}=\begin{bmatrix}
\C(Y_1,v_x)& \C(Y_1,v_y)  &  \C(Y_1,v_z)\\
\C(Y_2,v_x) & \C(Y_2,v_y)  & \C(Y_1,v_z)\\
\vdots & \vdots & \vdots \\
\C(Y_L,v_x)& \C(Y_L,v_y) &  \C(Y_L,v_z)
\end{bmatrix}.
\end{align}
\normalsize
Now, since the observations are assumed to be independent, we have $\C(Y_{l_1},Y_{l_2})=0$, for $l_1 \neq \l_2$. Hence it is straightforward to obtain \eqref{estimation_rule_csk_MMSE}.

\section{Proof of Lemma \ref{lemma_CR_est_L_performance}}
\label{AppendixProoflemma10}
We obtain the matrixes $J_D$ and $J_P$ in the BCR lower bound given in \eqref{cramer_rao_bound_BCR}, \eqref{cramer_rao_bound2}. The $(i,j)$-th entry of $J_P$ can be obtained as 
\begin{align}
&\{J_P\}_{i,j}=-\E_{v}[\frac{\partial^2}{\partial v_i \partial v_j}\ln(p_{x,y,z}(\bm{v}))]\\\nonumber
&=-\E_{v}[\frac{\partial^2}{\partial v_i \partial v_j}(\ln(p_{x}(v_x))+\ln(p_{y}(v_y))+\ln(p_{z}(v_z)))].
\end{align}
Hence, it is straightforward to obtain \eqref{CR_lower_bound2}. The $(i,j)$-th entry of the matrix $J_F(\bm{v})$ is obtained as
 \begin{align}
 \nonumber
\{J_F(\bm{v})\}_{i,j}=-\E_{\bm{Y}|v}[\frac{\partial^2}{\partial v_i\partial v_j}\big(\sum_{l=1}^{L}\ln(\Prob(Y_l|v))\big)]\\
=-\sum_{l=1}^{L}\E_{\bm{Y}|v}[\frac{\partial^2}{\partial v_i\partial v_j}\big(Y_l\ln(\lambda_l(\bm{v}))-\lambda_l(\bm{v})\big)],
 \end{align}
 which can be simplified to
 \small
 \begin{align}
 \nonumber
\{J_F(\bm{v})\}_{i,j}&=-\sum_{l=1}^{L}\E_{\bm{y}|v}\bigg[\big(\frac{\partial}{\partial v_i}\big[\frac{1}{\lambda_l(\bm{v})}.\frac{\partial \lambda_l(\bm{v})}{\partial v_j}\big]\big)(y_l-\lambda_l(\bm{v}))\\
&\quad +(\frac{1}{\lambda_l(\bm{v})}.\frac{\partial \lambda_l(\bm{v})}{\partial v_j})(-\frac{\partial \lambda_l(\bm{v})}{\partial v_i})\bigg].
 \end{align}
 \normalsize
 Now, since $\E_{Y_l|v}[Y_l]=\lambda_{1}(v\bm{d})$ and the second term is not related to $\bm{y}$, we have
 \begin{align}
\{J_F(\bm{v})\}_{i,j}=\sum_{l=1}^{L}\frac{1}{\lambda_l(\bm{v})}.\frac{\partial \lambda_l(\bm{v})}{\partial v_j}\frac{\partial \lambda_l(\bm{v})}{\partial v_i}.
 \end{align}
Hence, the $(i,j)$-th entry of the matrix $J_D=\E_{\bm{v}}[J_F(\bm{v})]$ is obtained as \eqref{CR_lower_bound}.

\begin{IEEEbiographynophoto}{Maryam Farahnak-Ghazani}
received the B.Sc. and M.Sc. degrees in Electrical engineering from Sharif University of Technology (SUT), Tehran, Iran, in 2014 and 2016, respectively. She is currently a Ph.D. candidate at the Department of Electrical Engineering, Sharif University of Technology, Tehran, Iran.
Her main research interests include nano-scale communication, molecular communication networks, information theory, and wireless communnication.
\end{IEEEbiographynophoto}
\vskip -1\baselineskip plus -1fil

\begin{IEEEbiographynophoto}{Mahtab Mirmohseni}
(S'06-M'12-SM'19) is an associate professor at Department of Electrical Engineering, Sharif university of Technology (SUT), Iran. She is also affiliated with the Information Systems and Security Laboratory (ISSL), Sharif University of Technology.
She received the B.Sc., M.Sc. and Ph.D. degrees from Department of Electrical Engineering, Sharif University, all in the field of Communication Systems in 2005, 2007 and 2012, respectively. She was a postdoctoral researcher at Royal Institute of Technology (KTH), Sweden, in the School of Electrical Engineering till Feb. 2014. She was the recipient of the Award of the national festival of the Women and Science (Maryam Mirzakhani Award), 2019, and also was selected as an exemplary reviewer for IEEE Transactions on Communications in 2016. Her current research interests include different aspects of information theory, mostly focusing on molecular communication and secure and private communication.
\end{IEEEbiographynophoto}
\vskip -1\baselineskip plus -1fil

\begin{IEEEbiographynophoto}{Masoumeh Nasiri-kenari}
received her B.Sc. and M.Sc. degrees in Electrical Engineering from Isfahan University of Technology, Isfahan, Iran, in 1986 and 1987, respectively, and her Ph.D. degree in Electrical Engineering from University of Utah, Salt Lake City, in 1993. From 1987 to 1988, she was a Technical Instructor and Research Assistant at Isfahan University of Technology. Since 1994, she has been with the Department of Electrical Engineering, Sharif University of Technology, Tehran, Iran, where she is now a Professor. \\
Dr. M. Nasiri-Kenari founded Wireless Research Laboratory (WRL) of the Electrical Engineering Department in 2001 to coordinate the research activities in the field of wireless communication. Current main activities of WRL are Energy Harvesting and Green communications, 5G and Molecular Communications. From 1999-2001, She was a Co-Director of Advanced Communication Research Laboratory, Iran Telecommunication Research Center, Tehran, Iran. She is a recipient of Distinguished Researcher Award and Distinguished Lecturer Award of EE department at Sharif University of Technology in years 2005 and 2007, respectively, a Research Chair on Nano Communication Networks from Iran National Science Foundation (INSF) and the 2014 Premium Award for Best Paper in IET Communications. She holds a Research Grant on Green Communication in Multi-Relay Wireless Networks for years 2015-2018, from Swedish Research Council. Since 2014, she serves as an Associate Editor of IEEE Trans. Communications.
\end{IEEEbiographynophoto}
\end{document}